\documentclass[draftclsnofoot, onecolumn, comsoc,12pt]{IEEEtran}

\usepackage{graphicx,epsfig}
\usepackage[noadjust]{cite}
\usepackage{mcite}
\usepackage{amsfonts,helvet}
\usepackage{fancyhdr}
\usepackage{threeparttable}
\usepackage{epsf,epsfig}
\usepackage{amsthm}
\usepackage{amsmath}
\usepackage{siunitx}
\usepackage{amssymb}
\usepackage{dsfont}
\usepackage{subfigure}
\usepackage{color}
\usepackage{algorithm}
\usepackage{algpseudocode}
\usepackage{algcompatible}
\usepackage{enumerate}
\usepackage{gensymb}
\usepackage{cancel}

\newtheorem{theorem}{Theorem}

\newtheorem{definition}{Definition}
\newtheorem{example}{Example}

\newtheorem{lemma}{Lemma}

\newtheorem{remark}{Remark}

\usepackage{eucal}

\begin{document}

\title{Inter-Operator Base Station Coordination in Spectrum-Shared Millimeter Wave Cellular Networks}

\author{Jeonghun~Park, Jeffrey~G.~Andrews, and Robert~W.~Heath Jr.
%\IEEEauthorblockA{Department of Electrical and Computer Eng.\\
%The University of Texas at Austin\\
%Austin, TX 78712 USA\\
%Email: \{jeonghun\}@utexas.edu}}
\thanks{J. Park, J. G. Andrews, and R. W. Heath Jr. are with the Wireless Networking and Communication Group (WNCG), Department of Electrical and Computer Engineering, 
The University of Texas at Austin, TX 78701, USA. (E-mail: $\left\{\right.$jeonghun, rheath$\left\}\right.$@utexas.edu, jandrews@ece.utexas.edu)

This research was supported in part by a gift from AT$\&$T Laboratories and by the National Science Foundation under Grant No. NSF-CCF-1514275.
}}

\maketitle \setcounter{page}{1} 

\begin{abstract}
We characterize the rate coverage distribution for a spectrum-shared millimeter wave downlink cellular network.  Each of multiple cellular operators owns separate mmWave bandwidth, but shares the spectrum amongst each other while using dynamic inter-operator base station (BS) coordination to suppress the resulting cross-operator interference. We model the BS locations of each operator as mutually independent Poisson point processes, and derive the probability density function (PDF) of the $K$-th strongest link power, incorporating both line-of-sight and non line-of-sight states.
Leveraging the obtained PDF, we derive the rate coverage expression as a function of system parameters such as the BS density, transmit power, bandwidth, and coordination set size. We verify the analysis with extensive simulation results. A major finding is that inter-operator BS coordination is useful in spectrum sharing (i) with dense and high power operators and (ii) with fairly wide beams, e.g., $30^{\circ}$ or higher.
\end{abstract}

\section{Introduction}
%The highly directional and adaptive antennas used in mmWave communication open up the possibility of uncoordinated sharing of spectrum licenses between commercial cellular operators.

%Coping with the dearth of the existing spectrum, use of millimeter wave (mmWave) in a cellular network is promising for the next-generation cellular architecture. 
Millimeter wave (mmWave) cellular networks improve conventional cellular data rates due to their large bandwidths \cite{rappa:access:13, roh:commmag:14, bai:twc:15}. The total amount of mmWave spectrum that is likely to be accessible to cellular operators in the near future, though, is a relatively small fraction of the total possible spectrum. Historically, operators acquire exclusive licenses in the spectrum \cite{boccardi:commag:16, shokri:jsac:16}, which further degrades the amount of spectrum that any particular mobile user can access. Given the novel interference-reducing features of millimeter wave systems, most notably directional beamforming and sensitivity to blocking, it may be preferable to pool and share spectrum licenses among multiple cellular operators \cite{boccardi:commag:16, shokri:jsac:16, gupta:tcom:16}. 
For example, \cite{gupta:tcom:16} showed that even uncoordinated sharing can increase the median rate. The favorable tradeoff observed in \cite{gupta:tcom:16} is that the bandwidth increase from spectrum sharing has a more significant (positive) impact on the rate of most users than the SINR degradation from the increased interference. This is a different tradeoff than in conventional cellular systems which do not benefit from highly directional beamforming: in such systems uncoordinated spectrum sharing is a losing proposition.

%In this paper, we consider inter-operator BS coordination as another enabler of spectrum sharing. As shown in \cite{gupta:tcom:16}, directional beamforming 

To obtain consistent rate gains from spectrum sharing in a wide variety of cellular network environments, directional beamforming may not be sufficient for interference suppression. For example, cellular operators can have different BS deployment densities, and users of a network with fewer BSs suffer when another operator's BSs become interferers, even accounting for highly directional beamforming. Unless such inter-operator interference is managed, the lower density operator is not incentivized to share spectrum with the other operator. Also, recent measurements in \cite{valenzuela:ctw:17} show that there is more scattering and dispersion in mmWave systems than commonly believed, especially for non line-of-sight (NLoS) paths. Thus, the actual interference can be much higher than what a simple sectored antenna model would predict, since significant interference could be received even from beams pointed out of main-lobe directions.

%Not super surprising, but one consequence is that such receivers do not see narrow beams from either their own or the interfering BSs.  Thus, the interference is much more dispersed.
%he key point is that even with narrow beams (narrower than 30%, the threshold in your paper), you get quite a bit of scattering and dispersion in NLOS and so much lower SINR than with a simple sectored model
%In addition, forming narrow beams in practice requires beam alignment overheads \cite{barati:twc:16}, which also erodes the gain of spectrum sharing.

%consider inter-operator BS coordination as another enabler of spectrum sharing in mmWave cellular networks. 
%In this paper, we investigate whether inter-operator BS coordination is necessary for spectrum sharing in mmWave cellular networks. 
The main goal of this paper is to characterize the prospective gain of inter-operator BS coordination in the context of spectrum sharing. Such coordination would reduce the interference. Although inter-operator BS coordination may currently seem impractical, it could be reasonable in future networks which are trending towards ever-increasing infrastructure aggregation \cite{boccardi:commag:16, kibilda:twc:17}. 
In the meantime, it provides a useful upper-bound on the possible gains from coordination.
%revealing it is a key enabler for general use of spectrum sharing.

%Specifically, we show that inter-operator BS coordination significantly increases the rate performance, by removing a few terms of signicant interference coming from different operator BSs. Based on this observation, we claim that inter-operator BS coordinaion is 
%We theoretically investigate the rate coverage performance spectrum-shared mmWave cellular network applied BS coordination, and 
%Specifically, removing a few terms of significant interference brings significant gains, 
%%Motivated by this, we investigate BS coordination applied spectrum sharing in this paper. %interference reduction via directionality 
%We claim that BS coordination complementary enabler for narrow directional beamforming, where one can assist if the other one cannot be used suitably. 

\subsection{Prior Work}
%Spectrum sharing is a well-studied subject in general, for example in the context of cognitive radios \cite{haykin:jsac:05, goldsmith:ieee:09, akyildiz:commmag:08}. 
Spectrum sharing is a well-studied subject in general, for example in the context of cognitive radios \cite{haykin:jsac:05, goldsmith:ieee:09}. 
We focus on mmWave spectrum sharing, which is much less studied. Following the earlier references \cite{boccardi:commag:16, shokri:jsac:16, gupta:tcom:16}, \cite{gupta:jsac:16} proposed a simple power control method to enhance the edge rate of primary users, who suffered in uncoordinated sharing. Specifically, secondary BSs decrease their transmit power such that their resulting interference is below threshold. In \cite{rebato:infocom:16}, not only spectrum, but also infrastructure and access sharing strategies were considered. In \cite{rebato:arxiv:16}, hybrid spectrum sharing was proposed, wherein the $20$ and $30$ GHz bands are used exclusively, while the $70 \;{\rm GHz}$ bands are shared. Users are jointly scheduled to one of these two bands depending on their SINRs, so that interference-limited users use $20/30 \;{\rm GHz}$ and noise-limited users use $70 \;{\rm GHz}$. This opportunistic sharing method shows some performance gain compared to a baseline sharing method. In \cite{li:crowncom:14}, an on-off spectrum sharing policy was proposed, where each operator allows the other operators to share the spectrum only if they incur a moderate level of interference. A similar approach was applied in WLAN systems \cite{feng:gcwkr:14}. In \cite{fund:sarnoff:16, fund:arxiv_economic:17}, an economic perspective of spectrum sharing in a mmWave cellular network was explored. A common aspect of prior work \cite{gupta:tcom:16, gupta:jsac:16, rebato:infocom:16, rebato:arxiv:16, li:crowncom:14, feng:gcwkr:14, fund:sarnoff:16, fund:arxiv_economic:17} is that they did not consider BS coordination between different operators. 
%One additional obstacle of cognitive radio is the requirement for continuous listening/sensing spectrum, which increases hardware complexity and power consumption. 

Inter-operator coordination in spectrum-shared mmWave cellular networks was discussed in \cite{boccardi:commag:16, shokri:jsac:16}. In \cite{boccardi:commag:16}, several network architectures that allow inter-operator coordination were presented, such as having a standardized core network interface. Alternatively, a new network entity called a \emph{spectrum broker} can be adopted for exchanging the information required for inter-operator coordination. The most closely related prior work is \cite{shokri:jsac:16}, where optimal cell association in spectrum-shared mmWave cellular networks employing inter-operator coordination was studied. A key difference in our work is an analysis of rate performance assuming random BS and user locations. 

%Now we clarify mathematical novelty of this paper by comparing to prior work. 
%Stochastic geometry has been used to model mmWave networks \cite{andrews:tcom:17}. 
%but also other spectrum-shared sub-$6$ GHz networks \cite{guo:twc:17, kibilda:twc:17}. Nevertheless, the analytical framework presented in the prior work \cite{andrews:tcom:17, guo:twc:17, kibilda:twc:17} cannot be applied in our scenario since mmWave networks have different path-loss characteristics compared to sub-$6$ GHz cellular networks. Specifically, the path-loss of line-of-sight (LoS) and non line-of-sight (NLoS) paths are different in mmWave cellular networks. There is prior work \cite{gupta:tcom:16, gupta:jsac:16} that incorporated this difference into the analysis, while BS coordination was not considered. Consequently, our case requires a new approach. 
%Since there is no prior work that analyzed BS coordination in a spectrum-shared mmWave cellular network. 
%For this reason, our case requires new approach. 
%BS coordination in mmWave cellular networks require a different approach, which will be explained later. For this reason, our case requires a new approach. 
%consider BS coordination in a mmWave cellular network, so that our case requires a new approach.

\subsection{Contributions}

%In this paper, we characterize the rate coverage distribution of spectrum-shared mmWave cellular networks applying inter-operator BS coordination. 
%We consider a cellular network comprising of multiple cellular operators, each of which owns separate mmWave bandwidth. 
%These operators share the spectrum among other operators. 
%For performance analysis of this network, we model each operator's BSs' locations by using mutually independent homogeneous Poisson point processes (PPPs). 
%We assume that each operator has different density and transmit power.
%With spectrum sharing, the users make use of increased bandwidth, while suffering from two kinds of interference, i.e., intra-operator and inter-operator interference. Dynamic BS coordination is applied to manage the interference. 
%Specifically, the BSs causing strong interference to a user form a coordination set. The coordination set can include different operators' BSs. The BSs in the set design their precoders to remove the interference inside the set, so that the intra-operator and also the intra-operator interference are mitigated.

In this paper, we characterize the rate coverage distribution of spectrum-shared mmWave cellular networks. We assume that inter-operator BS coordination is exploited to mitigate the interference from other operators. 
%With spectrum sharing, the users have two kinds of interference, i.e., intra-operator and inter-operator interference. 
Specifically, the $K$ strongest BSs of other operators are included in a coordination set. Subsequently, the BSs in the coordination set use precoders to remove the mutual interference in the set.  
%When the total coordination set includes different operators, the user employs inter-operator BS coordination. In this case, inter-operator interference is mitigated. When the coordination set includes only one operator, the user employs intra-operator BS coordination, where only intra-operator interference can be removed.
To characterize the gain of BS coordination, we derive the PDF of the link power corresponding to the  $K$-th {strongest} BS. We note that this is different from BS coordination in prior work \cite{nigam:tcom:14, lee:twc:15, li:tcom:15}, which assumed a single link state so that the $K$-th strongest BS is also the $K$-th closest BS. In our case, LoS and NLoS links are instead distinguished by their path-loss exponents and path-loss intercepts, so link power is determined not only by link distance, but also by the LoS/NLoS state. As a result, the $K$-th strongest BS may not be equal to the $K$-th closest BS. The derived PDF incorporates this feature. We also show that the obtained PDF reduces to the previous results \cite{lee:twc:15} when system assumptions are simplified. In this sense, the obtained PDF is more general than \cite{lee:twc:15, li:tcom:15}.

Leveraging the obtained PDF, we derive the rate coverage expression, which is a function of system parameters such as the density, transmit power, bandwidth, path-loss exponents and path-loss intercepts, and the BS coordination set size. The obtained expression indicates how the rate coverage performance is affected by the system parameters. 
For example, when the coordination set size increases, less interference remains, which leads to rate coverage improvement. 
%When the coordination set size decreases, the opposite happens. 
When there is no inter-operator BS coordination, there is no interference mitigation and this reduces the obtained expression to the previous results \cite{gupta:tcom:16}, which assumed uncoordinated spectrum sharing. In the simulation results, we verify the correctness of the obtained expressions. 

%This is a key difference with dynamic BS coordination in prior work \cite{lee:twc:15}. 
%To solve this problem, we derive the probability density function (PDF) of the $K$-th strongest BS' link power accounting for the LoS and NLoS states. We show that the obtained PDF reduces to previous results when system assumptions are simplified
%In this model, one major difficulty is that characterizing the $K$-th strongest BS is not straightforward.  
%Leveraging the PDF, we obtain a rate coverage expression. That is a function of system parameters such as the density, transmit power, bandwidth, path-loss exponents and constant gains, and the coordination set size. In the simulation results, we verify the correctness of the obtained expression. 

Our major findings from the analysis are as follows: (i) By using inter-operator BS coordination, spectrum sharing provides significant gains over uncoordinated case when sharing the spectrum with a dense and high power operator. (ii) Intra-operator BS coordination offers only marginal performance gain. (iii) Inter-operator BS coordination is more efficient when the beams are fairly wide, implying that they are complementary in the role of interference mitigation. 
In addition, we expect that inter-operator BS coordination is also valuable when there is sufficient scattering and dispersion.
%Summarizing the findings, we conclude that inter-operator BS coordination is valuable when sharing spectrum in mmWave networks. 

The paper consists of four main parts. We introduce the system models in Section II, we characterize the performance of BS coordination in a spectrum-shared mmWave cellular network in Section III, and we provide numerical results in Section IV. We conclude the paper in Section V. 
%The paper is organized as follows. Section II introduces the system models used in the paper. In Section III, the performance of BS coordination applied in a spectrum-shared mmWave cellular network is characterized. In Section IV, the numerical results are provided. Section V concludes the paper.

\section{System Model}
In this section, we introduce the system model and assumptions used in this paper. We first describe the network model based on stochastic geometry and the spectrum sharing model. Then, we explain the difference of LoS and NLoS states and how the typical user is associated with the BS incorporating LoS/NLoS BSs. Next, we illustrate the inter-operator BS coordination in detail. In the following subsection, we introduce the channel model and performance metrics. 

\subsection{Network and Spectrum Sharing Model}
We consider a downlink network comprising of $M$ cellular operators, all using mmWave bands. Focusing on operator $m$ for $m \in \CMcal{M} = \{1,...,M\}$, the locations of the BSs are modeled as a homogeneous PPP $\Phi_m = \{{\bf{d}}_i^{(m)}, i \in \mathbb{N}\}$ with density $\lambda_m$. 
%The locations of the BSs in two different operator are mutually independent. 
The BS locations of different operators are mutually independent, i.e., $\Phi_m$ and $\Phi_{m'}$ are independent for $m \neq m'$. 
Without loss of generality, we assume that $\left\| {\bf{d}}_i^{(m)}\right\| \le \left\| {\bf{d}}_j^{(m)}\right\|$ if $i<j$, so that ${\bf{d}}_1^{(m)}$ is the closest BS to the origin in operator $m$. 
%The density of each operator is assumed that $\lambda_1 \le \lambda_2$, so that the operator $1$ is has lower density. 
The transmit power of operator $m$ is denoted as $P_m$. 
All the operators are equipped with $N$ antennas and $N_{\rm RF} $ RF chains for $N_{\rm RF} \le N$, where a hybrid precoding method \cite{alk:twc:15, kulkarni:tcom:16} is used. 

Users are also distributed as a homogeneous PPP, $\Phi_{m}^{\rm (u)} = \{{\bf{u}}_i^{(m)}, i \in \mathbb{N}\}$ with density $\lambda_m^{\rm (u)}$. In all the operators, a single user is served from its associated BS. We assume that the density of users $\lambda_{m}^{\rm (u)}$ is far greater than $\lambda_m$ for all $m \in \CMcal{M}$, thereby there is no empty cell with high enough probability. Per Slivnyak's theorem \cite{baccelli:inria}, we henceforth focus on the typical user located at the origin denoted as ${\bf{o}}$. Without loss of generality, we assume that the typical user is in operator $1$. 

All the operators in the network share the spectrum among other operators. We assume that each operator owns separate bands, and denote that operator $m$'s bandwidth is $W_m$. In spectrum sharing, the typical user makes use of effective bandwidth $\sum_{m \in \CMcal{M}}W_m$. 

%We explain the relationship between the spectrum sharing set $\CMcal{M}$ and the BS coordination set $\CMcal{A}_m$. If $m \notin \CMcal{M}$, the typical user do not use BS coordination in the operator $m$, i.e., $\CMcal{A}_m = \emptyset$ since there is no interference coming from the operator $m$. 
%
%In this paper, we do not assume access or infrastructure sharing, so that the typical user cannot access with a BS of other operator $m \neq 1$.  

\subsection{Link State and Association Model}
Any link from a BS to the typical user is LoS or NLoS. Each state is represented by a state parameter $s$, where $s={\rm L}$ means a LoS link and $s = {\rm N}$ means a NLoS link. 
As in \cite{bai:twc:15}, the LoS/NLoS states are randomly determined depending on the link distance. Assuming a link between the typical user and an arbitrary BS located at ${\bf{d}}_i^{(m)}$ whose distance is $\left\| {\bf{d}}_i^{(m)}\right\| = r $, the link state is LoS (or $s = {\rm L}$) with the probability $p(r) = e^{-r/\mu}$, where $\mu$ is the average LoS length. The parameter $\mu$ is determined depending on the blockage density and the geometry \cite{bai:twc:15}. Under this setting, the state parameter $s$ is a random variable for each link. 
 
A LoS link and a NLoS link are different in their path-loss exponents and intercepts, denoted as $\alpha_{s}$ and $C_s$ for $s \in \{\rm L, N\}$. For example, considering a LoS link whose link distance is $r$, the corresponding path-loss is $C_{\rm L} r^{-\alpha_{\rm L}}$. For a NLoS link with the same distance, the corresponding path-loss is $C_{\rm N} r^{-\alpha_{\rm N}}$.  For ease of notation, we separate the total set $\Phi_m$ into a LoS BS set and a NLoS BS set depending on the corresponding links' states. The BS located at ${\bf{d}}_i^{(m)}$ is included in $\Phi_{m, \rm L}$ if the link between ${\bf{o}}$ and ${\bf{d}}_i^{(m)}$ is LoS, otherwise it is included in $\Phi_{m, \rm N}$. We note that $\emptyset = \Phi_{m, \rm L} \cap \Phi_{m, \rm N}$ and $\Phi_{m} = \Phi_{m, \rm L} \cup \Phi_{m, \rm N}$ for $m \in \CMcal{M}$. 

%In LoS and NLoS BSs in operator $1$, 
Among all the BSs including LoS and NLoS BSs in operator $1$, the typical user is associated with the strongest BS. Denoting $i_a$ as the associated BS index, we write 
\begin{align} \label{eq:asso_rule}
i_a = \mathop {\arg \max} \limits_{i \in \mathbb{N}} C_s \left\|{\bf{d}}_i^{(1)} \right\|^{-\alpha_s}.
\end{align}
We note that the associated BS can be changed depending on the link state variable $s$. 

%Similar to this, considering a link in cmWave bands, the path-loss is $C_{\rm \{L,cm\}} \left\| {\bf{d}}_i^k \right\|^{-\alpha_{\rm \{L, cm\}}}$ for the LoS case, and $C_{\rm \{N,cm\}} \left\| {\bf{d}}_i^k \right\|^{-\alpha_{\rm \{N,cm\}}}$ for the NLoS case. 

%Conventionally the path-loss becomes severe as the wavelength is smaller. For this reason, we define $C_{\rm C} > C_{\rm L} > C_{\rm N}$. 

\subsection{Base Station Coordination Model}
We first define a BS coordination set $\CMcal{A}_m$ for $m \in \{1,...,M\}$. 
%Note that each operator has its own coordination set. 
%The coordination set of each operator is formed, and then makes an union set. 
%The BSs included in the same coordination set cooperate each other to mitigate the interference to the typical user. The coordination set $\CMcal{A}_m$ only includes the operator $m$'s BSs and there is no coordination between different coordination sets. 
To form the coordination set $\CMcal{A}_m$, a dynamic clustering strategy is used, where the $K_m$ strongest BSs of operator $m$ are included in $\CMcal{A}_m$. For example, if the coordination set $\CMcal{A}_m = \{{\bf{d}}_{i_1}^{(m)},...,{\bf{d}}_{i_{K_m}}^{(m)}\}$ then the following satisfies.
\begin{align} \label{eq:clustering}
C_{s} \left\|  {\bf{d}}_{i_1}^{(m)}\right\|^{-\alpha_{s}} \ge C_{s} \left\|  {\bf{d}}_{i_2}^{(m)}\right\|^{-\alpha_{s}} \ge... \ge C_{s} \left\|  {\bf{d}}_{i_{K_m}}^{(m)}\right\|^{-\alpha_{s}},
\end{align}
and $C_{s} \left\|  {\bf{d}}_{i_{K_m}}^{(m)}\right\|^{-\alpha_{s}} \ge C_{s} \left\|  {\bf{d}}_{j}^{(m)}\right\|^{-\alpha_{s}}$ for all $j \in \mathbb{N} \backslash \CMcal{A}_m$. We note that the coordination set $\CMcal{A}_m$ only includes operator $m$'s BSs, i.e., $\CMcal{A}_m \subseteq \Phi_{m}$. 
%Other operator's BSs are not a member of $\CMcal{A}_m$. 
%This means that the $K$ strongest BSs in operator $m$ are included in the set $\CMcal{A}_m$ and $\CMcal{A}_m \subseteq \Phi_{m}$. 
The other operator's coordination set $\CMcal{A}_{m'}$, $m \neq m'$ is formed by the similar way. For operator $1$, it is always true that $\left| \CMcal{A}_1\right| \ge 1$ since the typical user is associated with the strongest BS by the association rule \eqref{eq:asso_rule}. 
%For the operator $1$, we have $\left| \CMcal{A}_1\right| \ge 1$ since the typical user is always associated with the strongest BS in the operator $1$. 

It is worthwhile to note that members of $\CMcal{A}_m$ change depending on the each link's state since the link power depends on link state. Unlike our case, assuming there exists a single link state as in \cite{lee:twc:15, li:tcom:15}, $\CMcal{A}_m$ is fixed as $\CMcal{A}_m = \{{\bf{d}}_1^{(m)}, ..., {\bf{d}}_{K_m}^{(m)}\}$, which is equal to a set of the $K_m$ closest BSs. In that case, the link power is solely determined by link distance. This is not the case in our setting since two link states, i.e., LoS/NLoS, are considered, making it a key difference from dynamic BS coordination in prior work \cite{lee:twc:15, li:tcom:15}.

Once each coordination set $\CMcal{A}_m$ for $m \in \CMcal{M}$ is formed, they make a total set $\CMcal{A}_{\rm total}$ by $\CMcal{A}_{\rm total} = \bigcup_{m=1}^{M} \CMcal{A}_m$. 
By this formation, the total coordination set $\CMcal{A}_{\rm total}$ is able to include BSs of multiple different operators, i.e., $\CMcal{A}_{\rm total} \subseteq \bigcup_{m = 1}^{M} \Phi_m$. All the BSs included in $\CMcal{A}_{\rm total}$ use precoder to remove the mutual interference inside the set $\CMcal{A}_{\rm total}$. The interference cancellation process will be explained in detail later.
We note that the cardinality of each coordination set $\left| \CMcal{A}_m\right|$ indicates the coordination level of operator $m$. For example, assume that $\left|\CMcal{A}_m \right| = 0$, $m \neq 1$. Then, there is no BS coordination in operator $m$, so that the interference of operator $m$ is not mitigated. 
Assuming that $\left|\CMcal{A}_m \right| = 0$ for $m \in \CMcal{M} \backslash 1$, there is no BS coordination in other operators except $1$. Since the typical user is associated with the operator $1$'s BS, this means that only intra-operator BS coordination is used.
If $\left|\CMcal{A}_m \right| = 0$ for $\CMcal{M} \backslash 1$ and $\left| \CMcal{A}_1\right| = 1$, no intra- or inter-operator BS coordination is used, and the assumption becomes same to uncoordinated spectrum sharing \cite{gupta:tcom:16}. 
%We note that $\left| \CMcal{A}_1\right| \ge 1$ since the typical user is associated with the BS of operator $1$. 
We denote that $\left| \CMcal{A}_m\right|=K_m$, and $\left| \CMcal{A}_{\rm total} \right|=K_{\rm total}$. Since there is no intersection between $\CMcal{A}_{m}$ and $\CMcal{A}_{m'}$, we can write $K_{\rm total} = \sum_{m = 1}^{M} K_m$. For analytical simplicity, we assume $N_{\rm RF} = K_{\rm total}$, i.e., the number of equipped RF chains is equal to the total coordination set size. 

We explain the reason that the individual coordination sets $\CMcal{A}_m$, $m \in \CMcal{M}$ are formed first before making a total set $\CMcal{A}_{\rm total}$. Since each operator has different density and transmit power, directly forming a total coordination set by jointly considering multiple operators is complicated. For example, we should compute all the possibilities which operator's BS would be a member of the coordination set. When there are many operators, this causes too much analytical complexity. To avoid this, we form individual coordination sets $\CMcal{A}_m$ first and subsequently make $\CMcal{A}_{\rm total}$. As shown later, we are able to incorporate the effect of the BS coordination set size into a rate coverage expression with this way. 

We note that the considered inter-operator coordination assumes an ideal scenario that all the BSs in a network are able to join the coordination set. In practice, this can be restricted due to limited network connectivity. In this sense, our performance analysis indicates an upper-bound on the performance of inter-operator coordination.

\subsection{Channel Model}
%We use different channel models for in-cluster channels and out-of-cluster channels. Specifically, we adopt a single-path model for out-of-cluster channels, and a multiple-paths model for in-cluster channels. This is reasonable because the strongest BSs are included in a cluster, implying that several signal paths can alive until they arrive to the typical user.
%On the contrary, for out-of-cluster channels, most of the signal paths vanish while traveling to the typical user due to severe path-loss, so that only few signal path remain which fits a single-path model.
In this subsection, we describe the assumptions used for modeling the channel.
\begin{enumerate}
\item We assume a single-path channel. Due to its tractability, this assumption was implicitly used in prior work that investigated the mmWave network performance, e.g., \cite{bai:twc:15, sarabjot:jsac:15, gupta:tcom:16, gupta:jsac:16}.
\item We consider Rayleigh fading as small-scale fading. In the prior work \cite{gupta:tcom:16, gupta:jsac:16}, it was shown that Rayleigh fading does not change the major performance trends compared to more general fading such as Nakagami fading. 
\item We assume that all the users are equipped with a single omni-directional antenna. This assumption was used in prior work \cite{gupta:tcom:16, gupta:jsac:16} for analytical simplicity. Although users are equipped with multiple antennas to obtain directivity gain in practice, the key system insights can be adequately obtained with the single-antenna assumption as shown in \cite{gupta:tcom:16, gupta:jsac:16}. At the expense of analytical simplicity, the multiple antenna user case can be incorporated into the analysis. 
%Specifically, using multiple receive antennas in multi-user MIMO systems 
%The main differences as 
Specifically, using multiple receive antennas, directional beamforming can be used not only at the BSs and but also at the users, so that the interference signal can have various directivity gain depending on its direction \cite{bai:twc:15}. This makes the Laplace transform of the interference complicated compared to a single receive antenna case. 
\item Each BS aligns the beam direction to its associated user, where each user is independently selected. For this reason, the angle-of-departure (AoD) from the BS at ${\bf{d}}_i^{(m)}$ ($i \neq i_a$) to the typical user, denoted as $\theta_i^{(m)}$, follows the independent uniform distribution. 
\item  We assume that $\rho_i^{(m)}$ indicates distance-dependent path-loss defined as $\rho_{i}^{(m)} = C_s \left\|{\bf{d}}_i^{(m)} \right\|^{-\alpha_s}$, the small-scale fading is captured in $\beta_{i}^{(m)}$, where $\left| \beta_{i}^{(m)}\right|^2 \sim \rm Exp(1)$ due to the Rayleigh fading assumption (the second assumption). 
\item ${\bf{a}}(\theta_i^{(m)}) \in \mathbb{C}^{N}$ is the array response vector corresponding to the AoD $\theta_i^{(m)}$.
\end{enumerate} 
%transmitted in/outside of BS's main beam, or
%1) transmitted in the BS's main beam and received in the user's main beam, 2) transmitted outside of the BS's main beam and received in the user's main beam, 3) transmitted in the BS's main beam and received outside of the user's main beam, and 4) transmitted outside of the BS's main beam and received outside of the user's main beam. We note that in the single receive antenna case, no directional beamforming is exploited in the users, so the received signals can be separated into only two cases: 1) transmitted in the BS's main beam and 2) transmitted outside of the BS's main beam. For this reason, the Laplace transform of the interference becomes complicated in the multiple receive antennas case. 
%For example, considering the multiple antenna user case, the received signal obtains additional antenna gains from a user side when the beam direction between the user and the BS is matched. 
With the enumerated assumptions, the channel vector between a BS at ${\bf{d}}_i^{(m)} $ and the typical user is written as 
\begin{align}
 {\bf{h}}_{i}^{(m)}  = \sqrt{\rho_i^{(m)}} \beta_i^{(m)} {\bf{a}}(\theta_i^{(m)}). 
\end{align}
% the arr is approximated by using flat-top model. Specifically, assume that $ ({\bf{a}}_i^m(\theta_i^m))^*$ is used as an analog beamforming at a BS. Then, 
%When assuming an uniform linear array (ULA) with half-lambda spacing and an angle-of-departure of the BS is $\theta_i^m$, the array response vector is ${\bf{a}}_i^m = 1/\sqrt{N}\left[1 , e^{j \pi \sin(\theta_{i}^{m})},..., e^{j(N-1) \pi \sin(\theta_i^m)} \right]$.

%Next, we explain the in-cluster channel model. Except that a multiple-paths model is adopted for in-cluster channels, the same assumptions are used. A in-cluster channel vector between a BS at ${\bf{d}}_i^m \in \CMcal{A}_{\rm total}$ and the typical user is 
%\begin{align}
%\left( {\bf{h}}_i^m \right)_{\rm In.} = \sqrt{\rho_i^m}\sum_{\ell = 1}^{L} \beta_{i,\ell}^m {\bf{a}}(\theta_{i,\ell}^m),
%\end{align}
%where $L$ is the number of paths. The fading coefficient $\beta_{i,\ell}^{m}$ satisfies $\left| \beta_{i,\ell}^{m}\right| \sim {\rm Exp}(1)$ due to Rayleigh fading. 
%Since it is not a single path, the sectored antenna model cannot be exploited with the presented in-cluster channel models. 

\subsection{Signal Model}
We explain the interference mitigation process and the received signal at the typical user. Before proceeding, we assume that channel state information at transmitters (CSIT) required for the interference mitigation are known to the BSs perfectly. 
%Prior to this, a user can obtain the required CSI by using the existing channel estimation algorithm such as \cite{alk:jstsp:14, gao:cl:16}.F
To do this, each user obtains (assumed perfect) CSI by using an existing channel estimation algorithm such as \cite{alk:jstsp:14, gao:cl:16}, and then sends the obtained CSI to the associated BS via a (assumed perfect) feedback link. In practice, there will be both estimation and feedback errors. For this reason, our assumption gives an upper-bound on the achievable gains from BS coordination. Incorporating imperfect CSIT is good topic for future work.

For explanation, we suppose a coordination set $\CMcal{A}_{\rm total}$ with $\left| \CMcal{A}_{\rm total}\right| = K_{\rm total}$. Including the typical user, there are $K_{\rm total}$ number of users to be served from the BSs in $\CMcal{A}_{\rm total}$. Then, a BS in $\CMcal{A}_{\rm total}$ can form a multi-user channel to $K_{\rm total}$ users. With this channel, the precoder design for cancelling the mutual interference is basically equivalent to the two-stage precoder presented in \cite{alk:twc:15, kulkarni:tcom:16}. 
Specifically, each BS in $\CMcal{A}_{\rm total}$ first makes its analog beamformer by matching with the AoDs to each user. For example, considering an arbitrary BS in $\CMcal{A}_{\rm total}$, it uses analog beamformer that ${\bf{A}} = \left[{\bf{a}}(\theta_1), {\bf{a}}( \theta_2), ..., {\bf{a}}( \theta_{K_{\rm total}}) \right] \in \mathbb{C}^{N \times N_{\rm RF}}$, where $\theta_1$ is the AoD to the associated user and $\theta_2, ..., \theta_{K_{\rm total}}$ are the AoDs to the other users in $\CMcal{A}_{\rm total}$. 
This is feasible since each BS is equipped with $N_{\rm RF} = K_{\rm total}$ number of RF chains. 
Denoting ${\bf{h}}_1 \in \mathbb{C}^{N}$ as the channel vector to the associated user and ${\bf{h}}_2, ..., {\bf{h}}_{K_{\rm total}} \in \mathbb{C}^N$ as the channel vectors to the other users, the effective channel after the precoding matrix ${\bf{A}}$ is written as
\begin{align}
{\bf{H}}^* {\bf{A}} &=  \left[\begin{array}{c} { {\bf{h}}_1^*} \\
{{\bf{h}}_2^*} \\ {\vdots} \\ {{\bf{h}}_{K_{\rm total}}^*} \end{array} \right] \cdot \left[{\bf{a}}(\theta_1), {\bf{a}}( \theta_2), ..., {\bf{a}}( \theta_{K_{\rm total}}) \right] 
=  \left[\begin{array}{c} {{\bf{g}}_1^*} \\
{{\bf{g}}_2^*} \\ {\vdots} \\ {{\bf{g}}_{K_{\rm total}}^*} \end{array} \right],
\end{align}
where $ {\bf{g}}_1^*$ is the desired effective channel vector corresponding to the associated user, while $ {\bf{g}}_k^*$ for $k = 2, ..., K_{\rm total}$ are the effective interfering channel vectors to other users. 
%Note that we abuse the notations for compactly presenting the channel matrix. 
%In \eqref{eq:effec_ch}, 
%In \eqref{eq:effec_ch}, we observe two things. First, the directivity gain is approximated by using the sectored antenna model. 
%Second, the diagonal terms of the channel matrix obtains the main-lobe gain because the BS matches the analog beamformer to each user's AoD. 
%For the other terms, the directivity gain can be $G$ or $g$ depending on the users' relative geometry. 
Now the BS designs its digital precoder ${\bf{v}}$ to satisfy the ZF criterion, i.e., ${\bf{g}}_2^*{\bf{v}} = 0$, ..., ${\bf{g}}_{K_{\rm total}}^*{\bf{v}} = 0$. Since $N_{\rm RF} = K_{\rm total}$, such a ${\bf{v}}$ can be found with high probability. 
% Since the channel coefficients are drawn from a complex Gaussian random variable, the digital beamformer ${\bf{v}}$ is found with high enough probability. 
Using ${\bf{v}}$, the interference to other users is completely removed.  
After multiplying with ${\bf{v}}$, the modified desired channel is written as ${\bf{g}}_1^*{\bf{v}} = \sqrt{\rho_1} \tilde G \tilde \beta$, where $\tilde G$ is the modified beamforming gain and $\tilde \beta$ is the modified fading after cancelling the interference. 
We note that $\tilde G < G$ where $G$ is the full beamforming gain since ZF decreases the beamforming gain. In general, $\tilde G$ is determined by the users' relative geometry in the coordination set. In this paper, we fix this intra-cluster users' geometry so that we treat $\tilde G$ as a deterministic variable. This is similar to the approach in \cite{lee:twc:15}. In this way, we observe how much BS coordination gain is obtained with specific intra-cluster geometry. 
One heuristic way to calculate $\tilde G$ is using simulations. For example, we drop BSs and users according to PPPs, make each BS's precoder, and then calculate the instantaneous value of $\tilde G$ for the typical user. By repeating this process, we can obtain the average value of $\tilde G$.
Analyzing $\tilde G$ rigorously is interesting future work.
%the modified beamforming gain $\tilde G$ is generally determined by the users' relative geometry in the coordination set, 
%By treating this as a deterministic parameter, we as in \cite{lee:twc:15}. 
%we are able to understand how the coordination gain changes depending on the users' intra-cluster geometry. 

%Unfortunately, analyzing the exact value of $\tilde G$ is difficult in general since it is changed by the users' relative geometry in $\CMcal{A}_{\rm total}$. Since we are interested in characterizing the network performance, analyzing the modified gain rigorously is beyond our scope. For this reason, we assume that $\tilde G$ is fixed. 
% as $pG$, where $0 < p \le 1$ means the relative gain compared to the ideal main-lobe gain $G$. 
%When $p$ is small, it means that the user's directivity gain decreases due to the interference mitigation. This can occur when there exist other users within the main-lobe intended to the associated user. 
%On the contrary, when $p$ is large, the desired user obtains the nearly ideal directivity gain. This can occur when we select the users so that there is only one user within each main-lobe. Later, we derive a rate coverage probability as a function of $p$ and show how this affects the performance of BS coordination. 

For analytical tractability, we approximate the out-of-cluster channel's beamforming gain by using a sectored antenna model while neglecting the dependence on ZF. 
%Further, we neglect the effect of ZF on interfering links and dependence in pMU and Iu through {wu} for tractability
The same approximation is used in \cite{kulkarni:tcom:16}, which showed that this approximation is reasonable in a mmWave cellular network. 
The sectored antenna model was used in prior work \cite{bai:twc:15, renzo:twc:15, sarabjot:jsac:15, gupta:tcom:16, gupta:jsac:16} for simplifying the analysis while capturing the directivity gain. For example, we assume an out-of-cluster BS located at ${\bf{d}}_i^{(m)} \notin \CMcal{A}_{\rm total}$ and it uses analog beamformer ${\bf{a}}(\theta)$, which is the array response vector corresponding to the AoD $\theta$. Then, the directivity gain to the typical user is ${\bf{a}}^*(\theta_i^{(m)}) {\bf{a}}(\theta)$. Using the sectored antenna model, this directivity gain is approximated as 
\begin{align} \label{eq:sectored_model}
\left| {\bf{a}}^*(\theta_i^{(m)}) {\bf{a}}(\theta) \right| \approx G(\theta_i^{(m)}, \theta) = \left\{\begin{array}{ll} G = \frac{2\pi - (2\pi - \frac{2\pi}{N})\epsilon}{\frac{2\pi}{N}},& \text{if $\left| \theta_i^{(m)} - \theta \right|\ \le \frac{\pi}{N}$,}   \\ g = \epsilon,&   \text{if $\left| \theta_i^{(m)} - \theta \right| > \frac{\pi}{N}$.} \end{array} \right.
\end{align} 
In \eqref{eq:sectored_model}, $G$ and $g$ indicate the main-lobe gain and the side-lobe gain, respectively. We note that the values of $G$ and $g$ are adopted from \cite{shokri:tcom:15}. 
%In general, $\epsilon \ll 1$ which implies $G \gg g$. 
%We note that with the sectored antenna model, analog beamformer is restricted to an array response vector. 
As mentioned in the previous section, we note that there can be quite a bit of scattering and dispersion in mmWave systems, especially for NLoS paths \cite{valenzuela:ctw:17}. For this reason, real channel environments can be far from the used sectored antenna model. Specifically, the actual SINR may be lower than the SINR obtained in the sectored antenna model. Nevertheless, we use the sectored antenna model in this paper for analytical tractability as in prior work \cite{bai:twc:15, gupta:tcom:16, kulkarni:tcom:16}. Incorporating more realistic channel models is future work.

With the sectored antenna model, we define $\tilde G = pG $, where $0 < p \le 1$ means the relative gain compared to the ideal main-lobe gain $G$. 
When $p$ is small, it means that the desired channel's directivity gain significantly decreases due to the interference mitigation. This occurs when there exist many other users within the main-lobe intended to the typical user. 
On the contrary, when $p$ is large, the desired user obtains the nearly ideal directivity gain. This can occur when the users are selected so that there is only one user within each main-lobe. Later, we derive a rate coverage probability as a function of $p$ and show how this affects the performance of BS coordination.

After cancelling the interference in the set $\CMcal{A}_{\rm total}$, the received signal at the typical user is 
\begin{align} \label{eq:rsig}
y =& \sqrt{P_1 C_s pG} \tilde \beta_{i_a}^{(1)} \left\|{\bf{d}}_{i_a}^{(1)} \right\|^{-\alpha_s/2} x_{i_a}^{(1)} + \sum_{m \in \CMcal{M}} \sum_{i \in \mathbb{N} \backslash \CMcal{A}_m}  \sqrt{P_m C_s G(\theta_{i}^{(m)})} \tilde \beta_{i}^{(m)} \left\|{\bf{d}}_{i}^{(m)} \right\|^{-\alpha_s/2} x_i^{(m)} + z,
\end{align}
where $P_m$ is the transmit power of operator $m$, $\tilde \beta_{i}^{(m)}$ is a modified channel coefficient where $\left|\tilde \beta_{i}^{(m)} \right|^2 \sim {\rm Exp}(1)$, $\sqrt{C_s} \left\| {\bf{d}}_i^{(m)} \right\|^{-\alpha_s/2}$ is link-state dependent path-loss, $x_i^{(m)}$ is an information symbol sent from the BS $i$ in operator $m$, $G(\theta_{i}^{(m)})$ is the approximated directivity gain of a out-of-cluster BS, and $z \sim \mathcal{CN}(0, \sigma^2)$ is additive white Gaussian noise.
%It is worthwhile to note that the 
%In this sense, the performance based on the received signal model \eqref{eq:rsig} indicates a lower bound. 
%where the antenna gain of the desired link between the BS $i_a$ and the typical user is $G(\theta_{i_a}^1) = G_1$ since their directions are perfectly aligned. We note that $ {\bf{h}}_{i}^m ({\bf{v}}_{i}^m)^* \in \mathbb{C}$ is an exponential random variable with unit mean due to the property of ZF. 
With the received signal \eqref{eq:rsig}, we define the instantaneous SINR and the rate. We assume that the power of an information symbol is normalized as $\mathbb{E}\left[\left| x_i^{(m)} \right|^2\right] = 1$ for $i \in \mathbb{N}$ and $m \in \CMcal{M}$. Then the instantaneous SINR is given as 
\begin{align} \label{eq:sinr_ns}
{\rm SINR} = \frac{P_1 pG \left| \tilde \beta_{i_a}^{(1)} \right|^2 C_s \left\| {\bf{d}}_{i_a}^{(1)} \right\|^{-\alpha_s}}{\sigma^2 + \sum_{m \in \CMcal{M}} \sum_{i \in \mathbb{N} \backslash \CMcal{A}_m}  {P_m G(\theta_{i}^{(m)})} \left| \tilde \beta_{i}^{(m)} \right|^2 C_s\left\|{\bf{d}}_{i}^{(m)} \right\|^{-\alpha_s} },
\end{align}
where the noise power $\sigma^2 = N_0 \sum_{m \in \CMcal{M}} W_m$ with $N_0 = -174 {\rm dBm/Hz}$. 
%Since the out-of-set BSs' main lobe beam directions are determined independently to the typical user, the probability that the typical user is within the main lobe is $1/N$. 
%The noise power is $\sigma_{\rm NS}^2 = N_0 W_{1}$, where $N_0$ is the thermal noise power spectral density defined as $N_0 = -174 {\rm dBm/Hz}$. 
%With the defined ${\rm SINR}_{\rm NS}$ \eqref{eq:sinr_ns}, the rate coverage when not sharing the spectrum is also defined as
The rate coverage probability is define as 
\begin{align} \label{eq:rcov}
R_{\rm cov} ( \gamma) = \mathbb{P} \left[\left( \sum_{m \in \CMcal{M}} W_m \right) \log_2 \left(1 + {\rm SINR} \right) > \gamma\right],
\end{align}
where $\gamma $ is the rate threshold. 

%\begin{figure}[t]
%\centering
%$\begin{array}{cc}
%{\resizebox{0.56\columnwidth}{!}
%{\includegraphics{uncoord_sharing.pdf}}} \\
%\mbox{(a)}\\
%{\resizebox{0.67\columnwidth}{!}
%{\includegraphics{coord_sharing.pdf}}}  \\
% \mbox{(b)} 
%\end{array}$
%\caption{Diagram for comparison between uncoordinated spectrum sharing (a) vs. BS coordination applied spectrum sharing (b). }
%   \label{fig:diagram}
%\end{figure} 

\section{Performance Characterization}

In this section, we analyze the rate coverage performance of coordinated spectrum sharing. 
%As a stepping stone for the analysis, we first 
%The main difficulty that the existing analytical framework cannot be applied in our case is that 
A key ingredient in performance characterization is the distribution of the link power in the coordination set. For example, assuming that a coordination set size is $K_m$ in operator $m$, we need to characterize the $K$-th strongest link power $C_{s}\left\| {\bf{d}}_{i_{K}}^{(m)} \right\|^{-\alpha_{s}}$. This is because the interference comes outside of the coordination set, so the link power $C_{s}\left\| {\bf{d}}_{i_{K}}^{(m)} \right\|^{-\alpha_{s}}$ serves as a  protection boundary of the typical user. 
Unfortunately, it is not straightforward to obtain the $K$-th strongest link power distribution in a mmWave network. 
%The main difficulty in characterization is that the PDF of the $K$-th strongest BS's link power. This is because 
The main reason is that any link has one of the two states (LoS or NLoS) whereby the link power is determined by not only the distance, but also the link state. For this reason, the BS ordering based on the distance may not be equivalent to the BS ordering based on its link power. For instance, the $K$-th closest BS may be equal to the $K$-th strongest BS; so that the distribution of the $K$-th closest BS's distance \cite{haenggi:tit:05} cannot be applied in mmWave networks. To resolve this, we first find the $K$-th strongest link power distribution. Leveraging the derived distributions, we obtain the rate coverage probability subsequently. 
%so that the $K$-th closest BS does not mean the last member of the coordination set. 
%the existing framework \cite{haenggi:tit:05} is hard to be applied directly. 
%To resolve this, we first 

\subsection{Link Power Distribution}
 
%We first define a link power $C_{s} \left\|{\bf{d}}_i^m \right\|^{-\alpha_{s}}$
%We first derive the following lemma that characterize the $K$-th strongest BS's link power. 
In this subsection, we obtain the link power distributions required for charactering the rate performance. Specifically, we obtain two PDFs: the PDF of the $K$-th strongest link power and the joint PDF of the strongest and the $K$-th strongest link power. The following lemma is the first main result of this subsection. 

\begin{lemma} \label{lem:PDF_K}
Assume a generic mmWave network $\Phi$ with density $\lambda$ and the LoS probability $p(r)$. The PDF of the $K$-strongest BS's link power is 
%\begin{align} \label{eq:pdf_K}
%f_{T_K}(t) = e^{-\Lambda_{\rm L}(t) } e^{-\Lambda_{\rm N}(t)} \left(\sum_{k=0}^{K-1}\frac{\Lambda_{\rm L}(t)^k}{k!} \frac{\Lambda_{\rm N}(t)^{K-k-1}}{(K-k-1)!} \right) \cdot \left(-\Lambda_{\rm L}'(t) - \Lambda_{\rm N}'(t) \right),
%\end{align}
\begin{align} \label{eq:pdf_K}
f_{T_K}(t) = e^{-\Lambda_{\rm L}(t) } e^{-\Lambda_{\rm N}(t)}  \left(-\Lambda_{\rm L}'(t) - \Lambda_{\rm N}'(t) \right) \frac{\left(\Lambda_{\rm L}(t) + \Lambda_{\rm N}(t) \right)^{K-1}}{(K-1)!},
\end{align}
where $\Lambda_{\rm L}(t)$ and $\Lambda_{\rm N}(t)$ is the intensity measures of the LoS and NLoS BS defined as
\begin{align}
\Lambda_{\rm L}(t) = &  2\pi \lambda \int_{0}^{\left(\frac{C_{\rm L}}{t} \right)^{\frac{1}{\alpha_{\rm L}}}} p(x) x {\rm d}x,
\end{align}
and
\begin{align}
\Lambda_{\rm N}(t) =  2\pi \lambda \int_{0}^{\left(\frac{C_{\rm N}}{t} \right)^{\frac{1}{\alpha_{\rm N}}}} (1-p(x)) x {\rm d}x.
\end{align}
\end{lemma}
\begin{proof}
%We first denote the LoS BS set as $\Phi_{\rm L}$ and the NLoS BS set as $\Phi_{\rm N}$. 
See Appendix \ref{appen:lem1}.
\end{proof}

%\begin{remark} \normalfont
%Unlike a homogeneous PPP case \cite{haenggi:tit:05}, it is hard to obtain the distance distribution directly. This is because the link power between the LoS and the NLoS link are different even if their distances are same each other. For this reason, we rather use link power applied equivalently to the LoS and the NLoS links.
%\end{remark}

\begin{definition} \label{def:pdf_K} \normalfont 
For convenience, we define the PDF in \eqref{eq:pdf_K} corresponding to operator $m$ as $f^{(m)}(t)$, obtained by substituting $\lambda \leftarrow \lambda_m$ and $K \leftarrow K_m $ in a generic function \eqref{eq:pdf_K}. 
\end{definition}

\begin{figure}[!t]
\centerline{\resizebox{0.64\columnwidth}{!}{\includegraphics{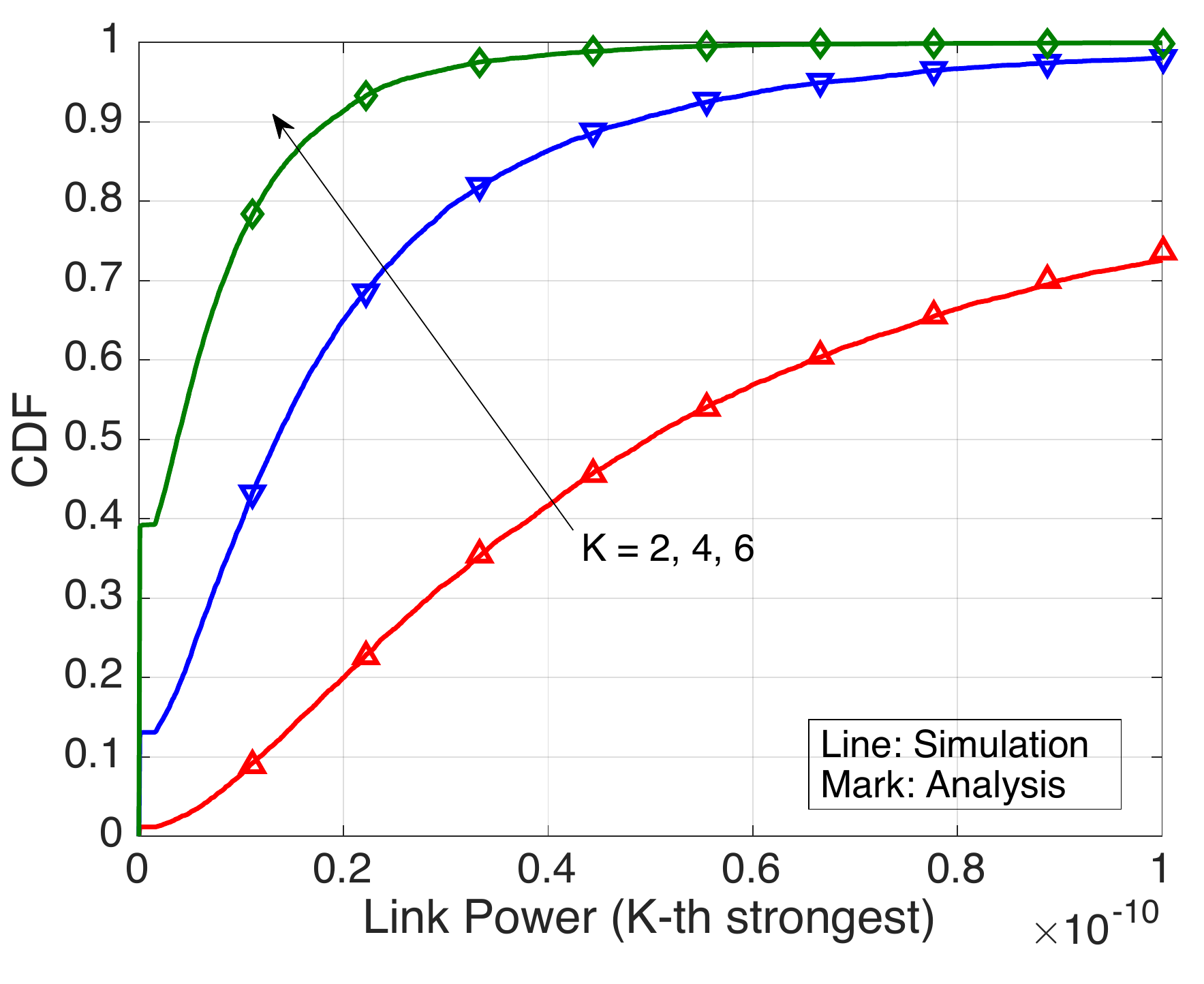}}}     
\caption{The comparison between the analytical CDF \eqref{eq:cdf_K_main} vs. the CDF obtained by numerical simulations. The assumed system parameters are as follows: $\alpha_{\rm L} = 2$, $\alpha_{\rm N} = 4$, $C_{\rm L} = -60{\rm dB}$, $C_{\rm N} = -70 {\rm dB}$, $\lambda = 5 \times 10^{-5}$, and $\mu = 144$.}
 \label{fig:K_cdf_verify}
\end{figure}

We note that Lemma \ref{lem:PDF_K} is a key ingredient to characterize the performance of BS coordination in mmWave cellular networks. We verify the obtained PDF by comparing the numerically obtained CDF in Fig.~\ref{fig:K_cdf_verify}. The analytical CDF is obtained by using Lemma \ref{lem:PDF_K} as follows
\begin{align} \label{eq:cdf_K_main}
F_{T_K}(t) = \int_{T_K = 0}^{t} f_{T_K}(t) {\rm d} t.
\end{align}
The system parameters are presented in the caption of the figure. As shown in Fig.~\ref{fig:K_cdf_verify}, the analytical results are perfectly matched with the simulation results, guaranteeing that our derivation is correct. 

The obtained PDF \eqref{eq:pdf_K} is general in the sense that it is applicable in different network scenarios, e.g., a sub-6 GHz cellular network or a non-cooperative mmWave network. Specifically, PDF \eqref{eq:pdf_K} reduces to the existing PDF by simplifying the network conditions. We study this in the following examples. 

\begin{example} \label{example:1} \normalfont
Setting $C_{\rm L} = C_{\rm N} = C$ and $\alpha_{\rm L} = \alpha_{\rm N} = \alpha$, the corresponding network boils down to a conventional sub-6 GHz cellular network, where only a single link state exists. In this case, the $K$-th strongest BS is equal to the $K$-th closest BS. Further, since there is a single link state, the state transition probability $p(r)$ has no meaning. For simplicity, we set $p(r) = 1/2$. Then, the intensity measure $\Lambda_{\rm L}(t) = \Lambda_{\rm N}(t) = \Lambda(t) = {\pi \lambda}/{2} \left(t/C \right)^{-{2}/{\alpha}} $. Accordingly, the PDF $f_{T_K}(t)$ \eqref{eq:pdf_K} is 
\begin{align} \label{eq:k_strong_k_close}
f_{T_K}(t) =& e^{-\lambda \pi \left(\frac{t}{C} \right)^{-\frac{2}{\alpha}}} \left(\frac{\pi \lambda}{2}  \left(\frac{t}{C} \right)^{-\frac{2}{\alpha}}\right)^{K-1} \sum_{k=0}^{K-1} \frac{1}{k! (K-k-1)!} \left( \frac{2\pi \lambda}{\alpha C} \left(\frac{t}{C} \right)^{-1-\frac{2}{\alpha}}  \right) \nonumber \\
\mathop{=}^{}& e^{-\lambda \pi \left(\frac{t}{C} \right)^{-\frac{2}{\alpha}}} \left(\frac{\pi \lambda}{2}  \left(\frac{t}{C} \right)^{-\frac{2}{\alpha}}\right)^{K-1} \sum_{k=0}^{K-1} \frac{(K-1)!}{k! (K-k-1)!} \left( \frac{2\pi \lambda}{\alpha C} \left(\frac{t}{C} \right)^{-1-\frac{2}{\alpha}}  \right)  \frac{1}{(K-1)!} \nonumber \\
\mathop{=}^{(a)}& e^{-\lambda \pi \left(\frac{t}{C} \right)^{-\frac{2}{\alpha}}} \left(\frac{\pi \lambda}{2}  \left(\frac{t}{C} \right)^{-\frac{2}{\alpha}}\right)^{K-1} 2^{K-1} \left( \frac{2\pi \lambda}{\alpha t} \left(\frac{C}{t} \right)^{\frac{2}{\alpha}}  \right)  \frac{1}{\Gamma(K)},
\end{align}
where (a) comes from a binomial series. 
%Now we consider the CDF of $T_K$ again, $F_{T_K}(t) = \mathbb{P}\left[T_K < t \right]$. 
%We can present the $K$-th strongest BS's link power as $T_K = C R_K^{-\alpha}$, where $R_K$ is the $K$-th strongest BS's distance. 
Now let us substitute $t = C r^{-\alpha}$, then $F_{T_K}(C r^{-\alpha}) = \mathbb{P}\left[T_K < C r^{-\alpha} \right]$. Since the $K$-th strongest BS's link power is presented as $C R_K^{-\alpha}$, $F_{T_K}(C r^{-\alpha})$ is interpreted as the CCDF of the $K$-th closest distance. Thus, the PDF of the $K$-th closest distance is obtained as
\begin{align} \label{eq:bridge}
f_{R_K}(r) = -\frac{\partial F_{T_K}(C r^{-\alpha})}{\partial r} = \alpha C r^{-1-\alpha} f_{T_K}(C r^{-\alpha}).
\end{align}
From \eqref{eq:k_strong_k_close}, we obtain 
\begin{align}
f_{R_K}(r) = \frac{ 2 \left({\pi \lambda} r^2 \right)^{K}}{r \Gamma(K)} e^{-\lambda \pi r^2}.
\end{align}
%we put $t = C r^{-\alpha}$, which is a path-loss of a link whose distance is $r$. T
%\begin{align}
%f_{R_K}(r) = e^{-\lambda \pi r^2} 2 \left({\pi \lambda} r^2 \right)^{K}\left( \frac{1}{\alpha C} r^{\alpha}  \right)  \frac{1}{\Gamma(K)} \left|\frac{\partial t}{\partial r} \right|
%\end{align}
%\begin{align}
%f_{T_K}(t) = \frac{ 2 \left({\pi \lambda} r^2 \right)^{K}}{\alpha C r^{-\alpha} \Gamma(K)} e^{-\lambda \pi r^2}
%\end{align}
This is equal to the previous result \cite{haenggi:tit:05}.
%, which revealed that the $K$-th closest point's distance in a homogeneous PPP with density $\lambda$. 
\end{example}

\begin{example} \normalfont
Assuming $K=1$, the PDF \eqref{eq:pdf_K} is reduced to the PDF of the association link's power in a mmWave network. In this case, we observe that the corresponding PDF consists of two parts as follows 
\begin{align} \label{eq:}
f_{T_1}(t) = \underbrace{-\Lambda_{\rm L}'(t) e^{-\Lambda_{\rm L}(t) } e^{-\Lambda_{\rm N}(t)}}_{(a)}  \underbrace{- \Lambda_{\rm N}'(t)  e^{-\Lambda_{\rm L}(t) } e^{-\Lambda_{\rm N}(t)}}_{(b)}.
\end{align}
%We first note that Lemma \ref{lem:PDF_K} does not particularly assume LoS or NLoS association 
%Since Lemma \ref{lem:PDF_K} incorporates both cases of LoS and NLoS association, 
We now denote that $f_{T_1, {\rm L}}(t)$ ($ f_{T_1, {\rm N}}(t)$) as the PDF of the association link's power when the typical user is associated with a LoS (NLoS) BS. 
%where $f_{T_1, {\rm L}}(t)$ is the PDF of the association link's power when the typical user is associated with a LoS BS. 
By showing that $f_{T_1, {\rm L}}(t) = (a)$ and $f_{T_1, {\rm N}}(t) = (b)$, we claim that the obtained PDF \eqref{eq:pdf_K} incorporates both cases of LoS and NLoS association. The PDF $f_{T_1, {\rm L}}(t)$ is obtained as follows
%Denoting that $A_{\rm L}$ is the LoS association probability, $\frac{1}{A_{\rm L}}  f_{T_1, {\rm L}}(t)$ is the conditional PDF of the association link's power conditioned on the LoS association. We focus on investigating $f_{T_1, {\rm L}}(t)$. 
%We first define that the specific regions are $B_1 = \left(0, t - {\rm d} t \right]$, $B_2 = \left(t - {\rm d}t, t\right]$, and the power of the strongest LoS BS is $T_{1, {\rm L}}$ and the strongest NLoS BS is $T_{1, {\rm N}}$. With these, the PDF \eqref{eq:pdf_K} can be written as
\begin{align}
f_{T_1, L}(t) = \lim_{{\rm d}t \rightarrow 0} \frac{\mathbb{P}\left[\{N_{\rm L}(t - {\rm d}t, t) = 1\} \cap \{N_{\rm L}(t, \infty) = 0\} \cap \{N_{\rm N}(t - {\rm d}t, \infty) = 0\} \right]}{{\rm d} t}.
\end{align}
%where $N_{\rm L}(t_1, t_2)$ ($N_{\rm N}(t_1, t_2)$) is the number of LoS (NLoS) BSs whose link power is in the region $\left(t_1, t_2 \right]$. 
Due to the independence of a PPP, we write
\begin{align} \label{eq:example2_derv}
f_{T_1, L}(t) =&  \lim_{{\rm d}t \rightarrow 0} \frac{\mathbb{P}\left[N_{\rm L}(t - {\rm d}t, t) = 1\right] \cdot \mathbb{P}\left[ N_{\rm L}(t, \infty) = 0\right] \cdot \mathbb{P}\left[ N_{\rm N}(t - {\rm d}t, \infty) = 0 \right]}{{\rm d} t} \nonumber \\
&= \lim_{{\rm d}t \rightarrow 0} -\frac{   \Lambda_{\rm L}(t)- \Lambda_{\rm L}(t - {\rm d}t)}{{\rm d} t} e^{-(\Lambda_{\rm L}(t-{\rm d} t) - \Lambda_{\rm L}(t)) } \cdot e^{-\Lambda_{\rm L}(t)} e^{-\Lambda_{\rm N}(t - {\rm d}t)} \nonumber \\
& \mathop{=}^{(c)} -\Lambda'_{\rm L}(t) e^{-\Lambda_{\rm L}(t)} e^{-\Lambda_{\rm N}(t )} = (a),
\end{align}
where (c) follows the definition of differentiation. Similarly, we also can show that $f_{\rm N}(t) = (b)$. 
Next, we assume $t = C_{\rm L}r^{-\alpha_{\rm L}}$ in \eqref{eq:example2_derv}. Then, we have
\begin{align}
f_{T_1, {\rm L}}(C_{\rm L}t^{-\alpha_{\rm L}}) = \frac{1}{\alpha C_{\rm L}} r^{1+\alpha_{\rm L}} 2 \pi \lambda p(r) r e^{-2\pi \lambda \int_{0}^{r} p(x) x {\rm d} x} e^{-2 \pi \lambda \int_{0}^{\left(\frac{C_{\rm N}}{C_{\rm L}} r^{\alpha_{\rm L}} \right)^{\frac{1}{\alpha_{\rm N}} }} (1-p(x)) x {\rm d} x}.
\end{align}
Since $f_{R_1, {\rm L}}(r) = \alpha_{\rm L} C_{\rm L} r^{-1-\alpha_{\rm L}}$ as shown in \eqref{eq:bridge}, the PDF of the association BS's distance when the typical user is associated with a LoS BS is 
\begin{align}
f_{R_1, {\rm L}}(r) = 2 \pi \lambda p(r) r e^{-2\pi \lambda \int_{0}^{r} p(x) x {\rm d} x} e^{-2 \pi \lambda \int_{0}^{\left(\frac{C_{\rm N}}{C_{\rm L}} r^{\alpha_{\rm L}} \right)^{\frac{1}{\alpha_{\rm N}} }} (1-p(x)) x {\rm d} x},
\end{align}
which is equivalent with the result derived in \cite{bai:twc:15}. 
\end{example}

\begin{example} \label{example:LoSratio}\normalfont
In this example, we study the average ratio of LoS and the NLoS BSs in the coordination set $\CMcal{A}$. Without loss of generality, we assume that $\left| \CMcal{A} \right| = K$ and denote a set of LoS BSs in $\CMcal{A}$ as $\CMcal{A}^{\rm L}$, and a set of the NLoS BSs in $\CMcal{A}$ as $\CMcal{A}^{\rm N}$. We note that $\CMcal{A} = \CMcal{A}^{\rm L} \cup \CMcal{A}^{\rm N}$ and $\CMcal{A}^{\rm L} \cap \CMcal{A}^{\rm N} = \emptyset$. 
For the characterization of $\left| \CMcal{A}^{\rm L}\right|$, we first obtain the probability of $\left| \CMcal{A}^{\rm L}\right| = k \le K$ conditioned on that the $(K+1)$-th strongest link power is $t$, i.e., $T_{K+1} = t$. This conditional probability is 
%For this reason, the conditional probability of $\left| \CMcal{A}^{\rm L}\right| = k$ is equal to that there are $k$ LoS BSs and $K-k$ NLoS BSs 
%Under this condition, there are two possibilities that $\left|\CMcal{A}^{\rm L} \right| = k$: (i) a LoS BS has the weakest link power in $\CMcal{A}$, i.e., its link power is $t$, and there are $k-1$ LoS BSs and $K-k$ NLoS BSs whose link power is in $\left(t, \infty \right]$. (ii) a NLoS BS has the weakest link power in $\CMcal{A}$, i.e., its link power is $t$, and there are $k$ LoS BSs and $K-1-k$ NLoS BSs whose link power is in $\left(t, \infty \right]$. For $k \ge 1$, the probability of the first possibility is written as 
\begin{align}
\mathbb{P}\left[\left. \left|\CMcal{A}^{\rm L} \right| = k\right| T_{K+1} = t \right] 
\mathop{=}^{(a)} \frac{k! (K-k)!}{K!}\left(\frac{\Lambda_{\rm L}(t)}{\Lambda_{\rm L}(t) + \Lambda_{\rm N}(t)} \right)^k \left(\frac{\Lambda_{\rm N}(t)}{\Lambda_{\rm L}(t) + \Lambda_{\rm N}(t)} \right)^{(K-k)},
\end{align}
where (a) comes from that, in a PPP, the number of points inside a certain window follow the multinomial distribution if the total number in the corresponding window is given. Specifically, 
%conditioning on $T_{K+1} = t$, implies that there are $K$ points in the window $\left(t, \infty \right]$. 
when $T_{K+1} = t$, there exist exactly $K$ points in the window $\left(t, \infty\right]$. Then, the number of BSs whose link power is in $\left(t, \infty\right]$ follows the multinomial distribution with the corresponding probability $\Lambda_{\rm L}(t)/\left(\Lambda_{\rm L}(t) + \Lambda_{\rm N}(t) \right)$. Similar to this, the number of NLoS BSs whose link power is in $\left(t, \infty\right]$ also follows the multinomial distribution with the corresponding probability $\Lambda_{\rm N}(t)/\left(\Lambda_{\rm L}(t) + \Lambda_{\rm N}(t) \right)$. Then, by using the average of a multinomial random variable, the conditional average of $\left| \CMcal{A}^{\rm L}\right|$ is obtained as follows 
\begin{align}
\mathbb{E}\left[\left. \left| \CMcal{A}^{\rm L}\right| \right| T_{K+1} = t\right] = K \cdot \frac{\Lambda_{\rm L}(t)}{\Lambda_{\rm L}(t) + \Lambda_{\rm N}(t)}.
\end{align}
Marginalizing for $T_{K+1}$, the average number of LoS BSs in $\CMcal{A}$ is
\begin{align}
\mathbb{E}\left[\left| \CMcal{A}^{\rm L} \right| \right] = \int_{0}^{\infty} \frac{K \Lambda_{\rm L}(t)}{\Lambda_{\rm L}(t) + \Lambda_{\rm N}(t)} f_{T_{K+1}}(t) {\rm d} t.
\end{align}
The average ratio of LoS BSs in $\CMcal{A}$ is thereby $\mathbb{E}\left[\left| \CMcal{A}^{\rm L} \right| \right]/K$ and the average ratio of NLoS BSs in $\CMcal{A}$ is $1-\mathbb{E}\left[\left| \CMcal{A}^{\rm L} \right| \right]/K$. Later, we show that how the ratio is changed as the coordination set size increases. 
\end{example}

Next, we obtain the joint PDF of the strongest and the $K$-th strongest BS's link power, which is the second main result in this subsection. 

\begin{lemma} \label{lem:jointPDF_1_K}
The joint PDF of the strongest BS's link power $T_1$ and the $K$-th strongest BS's link power $T_K$ is 
%\begin{align} \label{eq:jointpdf_final}
%f_{T_1, T_K}(t_1, t_K) =& \Lambda'_{\rm L}(t_K) \Lambda'_{\rm L}(t_1) e^{-\Lambda_{\rm L}(t_K)} e^{-\Lambda_{\rm N}(t_K)} \left(\sum_{k = 1}^{K-1}\frac{\Lambda_{\rm L}(t_K, t_1)^{(k-1)}}{(k-1)!} \frac{\Lambda_{\rm N}(t_K, t_1)^{(K-1-k)}}{(K-1-k)!} \right) \nonumber \\
%& + \Lambda'_{\rm L}(t_K) \Lambda'_{\rm N}(t_1) e^{-\Lambda_{\rm L}(t_K)} e^{-\Lambda_{\rm N}(t_K)} \left(\sum_{k = 0}^{K-2}\frac{\Lambda_{\rm L}(t_K, t_1)^{k}}{k!} \frac{\Lambda_{\rm N}(t_K, t_1)^{(K-2-k)}}{(K-2-k)!} \right) \nonumber \\
%& + \Lambda'_{\rm N} (t_K) \Lambda'_{\rm L}(t_1) e^{-\Lambda_{\rm L}(t_K)} e^{-\Lambda_{\rm N}(t_K)} \left(\sum_{k = 1}^{K-1}\frac{\Lambda_{\rm L}(t_K, t_1)^{(k-1)}}{(k-1)!} \frac{\Lambda_{\rm N}(t_K, t_1)^{(K-1-k)}}{(K-1-k)!} \right) \nonumber \\
%& + \Lambda'_{\rm N} (t_K) \Lambda'_{\rm N}(t_1) e^{-\Lambda_{\rm L}(t_K)} e^{-\Lambda_{\rm N}(t_K)} \left(\sum_{k = 0}^{K-2}\frac{\Lambda_{\rm L}(t_K, t_1)^{k}}{k!} \frac{\Lambda_{\rm N}(t_K, t_1)^{(K-2-k)}}{(K-2-k)!} \right),
%\end{align}
\begin{align} \label{eq:jointpdf_final}
& f_{T_1, T_K}(t_1, t_K) \nonumber\\
=&  \left(\begin{array}{c}{ \Lambda'_{\rm L}(t_K) \Lambda'_{\rm L}(t_1) + \Lambda'_{\rm L}(t_K) \Lambda'_{\rm L}(t_1) } \\{+ \Lambda'_{\rm L}(t_K) \Lambda'_{\rm L}(t_1) + \Lambda'_{\rm L}(t_K) \Lambda'_{\rm L}(t_1)} \end{array} \right) \cdot e^{-\Lambda_{\rm L}(t_K)} e^{-\Lambda_{\rm N}(t_K)} \frac{\left(\Lambda_{\rm L}(t_K, t_1) + \Lambda_{\rm N}(t_K, t_1)\right)^{(K-2)}}{(K-2)!},
\end{align}
if $t_K < t_1$, while $f_{T_1, T_K}(t_1, t_K) = 0$ otherwise. The differential intensity measure $\Lambda_{\rm L}(t_K, t_1)$ is defined as
\begin{align} \label{eq:diff_measure}
\Lambda_{\rm L}(t_K, t_1) =& \Lambda_{\rm L}(t_K) - \Lambda_{\rm L}(t_1).
\end{align}
\end{lemma}
\begin{proof}
See Appendix \ref{appen:lem2}. 
\end{proof}

\begin{definition} \label{def:K_1}\normalfont
Similar to Definition \ref{def:pdf_K}, we define the joint PDF corresponding to operator $m$ as $f^{(m)}_{\rm Joint}(t_1, t_{K_m})$, obtained by substituting $\lambda \leftarrow \lambda_m$ and $K \leftarrow K_m$ in a generic function \eqref{eq:jointpdf_final}. 
\end{definition}

\begin{remark} \normalfont
Lemma \ref{lem:jointPDF_1_K} is required to calculate the interference statistic when BS coordination is applied in operator $1$, i.e., $K_1 > 1$. For operator $m \in \CMcal{M}\backslash 1$, only the PDF of the $K$-th strongest link power is needed, while the joint PDF of the strongest and the $K$-th strongest link power is required for operator $1$ since the typical user receives the desired data from the strongest BS in operator $1$. 
\end{remark}

\begin{example} \normalfont
%We show that the obtained joint PDF \eqref{eq:jointpdf_final} reduces to the previous result by assuming a conventional cellular network. 
As in Example \ref{example:1}, we show that the obtained joint PDF \eqref{eq:jointpdf_final} reduces to the previous result in a conventional sub-6 GHz cellular network setting. We assume that $C_{\rm L} =  C_{\rm N} = C$ and $\alpha_{\rm L} = \alpha_{\rm N} = \alpha$, so that links have a single state path-loss. Since the state transition probability $p(x)$ has no meaning in this setting, we set $p(x) = 1$ for simplicity. Under this assumption, the obtained joint PDF \eqref{eq:jointpdf_final} is
%\begin{align} \label{eq:relation_tr}
%f_{T_1, T_K}(t_1, t_K) = \left(\frac{2 \pi \lambda}{\alpha}\right)^{2}  \left(\frac{t_K}{C} \right)^{-\frac{2}{\alpha}-1} \left(\frac{t_1}{C} \right)^{-\frac{2}{\alpha}-1} e^{-\pi \lambda \left(\frac{t_K}{C }\right)^{-2/\alpha}} \frac{\left(\pi \lambda\left( \left(t_K/C \right)^{-2/\alpha} - \left(t_1 / C \right)^{-2\alpha}\right) \right)^{K-2}}{(K-2)!}.
%\end{align}
\begin{align} \label{eq:relation_tr}
f_{T_1, T_K}(t_1, t_K) = \left(\frac{2 \pi \lambda}{\alpha}\right)^{2}  \left(\frac{C^2}{t_1t_K} \right)^{\frac{2}{\alpha}+1}  e^{-\pi \lambda \left(\frac{t_K}{C }\right)^{-2/\alpha}} \frac{\left(\pi \lambda\left( \left(t_K/C \right)^{-2/\alpha} - \left(t_1 / C \right)^{-2\alpha}\right) \right)^{K-2}}{(K-2)!}.
\end{align}
Due to the similar reason presented in Example \ref{example:1}, the joint PDF $f_{R_1, R_K}(r_1, r_K)$, where $R_1$ is the closest BS's distance and $R_K$ is the $K$-closest BS's distance, is obtained as
\begin{align}
f_{R_1, R_K}(r_1, r_K) = (\alpha C)^2 r_1^{-1-\alpha} r_K^{-1-\alpha} \cdot f_{T_1, T_K}(Cr_1^{-\alpha}, Cr_K^{-\alpha}).
\end{align}
Plugging $t_1 = C r_1^{-\alpha}$ and $t_K = C r_K^{-\alpha}$ into $f_{T_1, T_K}(t_1, t_K)$, we have
\begin{align}
f_{T_1, T_K}(C r_1^{-\alpha},C r_k^{-\alpha}) = \left(\frac{2}{\alpha}\right)^{2} \left(\pi \lambda \right)^2 \left(r_K \right)^{2+\alpha} \left(r_1 \right)^{2+\alpha} e^{-\pi \lambda r_K^2} \frac{\left(\pi \lambda\left( r_K^2 - r_1^2\right) \right)^{K-2}}{(K-2)!}.
\end{align}
By using \eqref{eq:relation_tr}, the joint PDF $f_{R_1, R_K}(r_1, r_K)$ is
\begin{align}
f_{R_1, R_K}(r_1, r_K) = 4 \left(\pi \lambda \right)^K  r_1 r_K e^{-\pi \lambda r_K^2} \frac{\left( r_K^2 - r_1^2\right)^{K-2}}{(K-2)!},
\end{align}
if $r_1 \le r_K$ and $f_{R_1, R_K}(r_1, r_K) = 0$ if $r_1 > r_K$. This is exactly same with the previous result obtained in \cite{lee:twc:15}, 
\begin{align}
f_{R_1, R_K}(r_1, r_K) = \left\{\begin{array}{lc} {\frac{4 (\lambda \pi)^K}{(K-2)!} r_1 r_K (r_K^2 - r_1^2)^{K-2} e^{-\lambda \pi r_K^2}}, & {{\rm if} \; r_1 \le r_K },\\ {0}, & {{\rm otherwise}}. \end{array} \right.
\end{align}
\end{example}

\subsection{Rate Coverage Analysis}
In this subsection, we leverage the obtained PDFs to calculate the rate coverage probability. First, we characterize the Laplace transform of the interference coming from operator $m \in \CMcal{M}\backslash 1$ in the following lemma.
\begin{lemma} \label{lem:laplace_int}
We denote the interference coming from operator $m$ for $m \in \CMcal{M} \backslash 1$ as
\begin{align}
I_m =   \sum_{i \in \mathbb{N} \backslash \CMcal{A}_m}  {P_m G(\theta_{i}^{(m)})} \left| \tilde \beta_i^{(m)}\right|^2 C_s\left\|{\bf{d}}_{i}^{(m)} \right\|^{-\alpha_s}.
\end{align}
The Laplace transform of $I_m$ is 
\begin{align}
&\CMcal{L}_{I_m}(s) = \int_{T = 0}^{\infty} \exp\left(- \left[\int_{T}^{0} \frac{sP_m G t}{1 +sP_m Gt } \tilde \Lambda^{m} ({\rm d} t) + \int_{T}^{0} \frac{sP_m gt}{1 +sP_m gt } \tilde {\tilde \Lambda}^{(m)} ({\rm d} t) \right]  \right) f_{T_{K_m}}^{(m)}(T) {\rm d} T,
\end{align}
where 
\begin{align}
&\tilde \Lambda^{(m)}({\rm d} t) \nonumber \\
&= -\frac{1}{N} \left[\frac{2\pi \lambda_m }{\alpha_{\rm L} C_{\rm L}} p\left(\left(\frac{C_{\rm L}}{t} \right)^{\frac{1}{\alpha_{\rm L}}} \right) \left( \frac{C_{\rm L}}{t} \right)^{\frac{2}{\alpha_{\rm L}}+1} {\rm d} t + \frac{2\pi \lambda_m}{\alpha_{\rm N} C_{\rm N}} \left(1 - p\left(\left(\frac{C_{\rm N}}{t} \right)^{\frac{1}{\alpha_{\rm N}}} \right) \right) \left( \frac{C_{\rm N}}{t}\right)^{\frac{2}{\alpha_{\rm N}}+1} {\rm d} t\right]
\end{align}
and 
\begin{align}
&\tilde {\tilde \Lambda}^{(m)}({\rm d} t) \nonumber \\
&= -\left(1 - \frac{1}{N}\right) \left[\frac{2\pi \lambda_m }{\alpha_{\rm L} C_{\rm L}} p\left(\left(\frac{C_{\rm L}}{t} \right)^{\frac{1}{\alpha_{\rm L}}} \right) \left( \frac{C_{\rm L}}{t} \right)^{\frac{2}{\alpha_{\rm L}}+1} {\rm d} t + \frac{2\pi \lambda_m}{\alpha_{\rm N} C_{\rm N}} \left(1 - p\left(\left(\frac{C_{\rm N}}{t} \right)^{\frac{1}{\alpha_{\rm N}}} \right) \right) \left( \frac{C_{\rm N}}{t}\right)^{\frac{2}{\alpha_{\rm N}}+1} {\rm d} t\right].
\end{align}
The PDF $ f_{}^{(m)}(T)$ is defined in Definition \ref{def:pdf_K}. 
\end{lemma}
\begin{proof}
See Appendix \ref{appen:lem3}. 
\end{proof}

%\begin{remark} \normalfont
%Comparing to the conventional Laplace transform obtained in , 
%our result has analytical difference, while the final form is equivalent. 
%\end{remark}
Exploiting the obtained Laplace transform, we derive the rate coverage probability in the following theorem. 

\begin{theorem} \label{theo:rcov}
When $K_1 > 1$, the rate coverage probability is 
\begin{align}
R_{\rm cov}(\gamma) = \int_{T_1 = 0}^{\infty} \int_{T_{K_1} = 0}^{T_1} e^{-\frac{\tilde \gamma \sigma^2}{P_1 pG T_1 }}\CMcal{L}_{I_1| T_{K_1}}\left(\frac{\tilde \gamma}{P_1 pG T_1}\right) \cdot \prod_{m \in \CMcal{M} \backslash 1} \CMcal{L}_{I_m} \left(\frac{\tilde \gamma}{P_1 pG T_1} \right) f^{(1)}_{\rm Joint}(T_1, T_{K_1}) {\rm d} T_{Ke_1} {\rm d} T_1,
\end{align}
where $\tilde \gamma = 2^{\gamma/\left(\sum_{m \in \CMcal{M}} W_m \right)}-1$, $f^{(1)}_{\rm Joint}(T_1, T_{K_1})$ is the joint PDF for operator $1$, and the conditional Laplace transform for operator $1$ is 
\begin{align}
\CMcal{L}_{I_1 | T_{K_1}}(s) =  \exp\left(- \left[\int_{T_{K_1}}^{0} \frac{sP_1 G t}{1 +sP_1 Gt } \tilde \Lambda^{m} ({\rm d} t) + \int_{T_{K_1}}^{0} \frac{sP_1 gt}{1 +sP_1 gt } \tilde {\tilde \Lambda}^{(m)} ({\rm d} t) \right]  \right). 
\end{align}
When $K_1 = 1$, the rate coverage probability is
\begin{align}
R_{\rm cov}(\gamma) = \int_{T_{1} = 0}^{\infty} e^{-\frac{\tilde \gamma \sigma^2}{P_1 pG T_1 }}\CMcal{L}_{I_1| T_{K_1}}\left(\frac{\tilde \gamma}{P_1 pG T_1}\right) \cdot \prod_{m \in \CMcal{M} \backslash 1} \CMcal{L}_{I_m} \left(\frac{\tilde \gamma}{P_1 pG T_1} \right) f^{(1)}(T_1)  {\rm d} T_1.
\end{align}
\end{theorem}
%and the conditional Laplace transform is
%\begin{align}
%\CMcal{L}_{I_1 | T_K}(s) = 
%\end{align}
\begin{proof}
With $\tilde \gamma$, the rate coverage probability is written as
\begin{align} \label{eq:rate_condition}
R_{\rm cov}(\gamma) =& \mathbb{P}\left[H_{i_a}^{(1)} > \frac{\sigma^2 \tilde \gamma}{P_1 pG T_1} + \sum_{m \in \CMcal{M}} \frac{I_m \tilde \gamma}{P_1 pG T_1} \right] \nonumber \\
=& \mathbb{E}_{T_1, T_{K_1}} \left[e^{-\frac{\sigma^2 \tilde \gamma }{P_1 pG T_1}} \CMcal{L}_{I_1 | T_{K_1}}\left( \frac{\tilde \gamma}{P_1 pG T_1}\right) \prod_{m \in \CMcal{M}\backslash 1} \CMcal{L}_{I_m}\left(\frac{\tilde \gamma}{P_1 pG T_1} \right) \right].
\end{align}
We obtain the conditional $\CMcal{L}_{I_1}(s)$ by using the similar way to Lemma \ref{lem:laplace_int}. One difference is that $T_1$ and $T_{K_1}$ in operator $1$ are correlated, so we marginalize them later at once 
\begin{align}
\CMcal{L}_{I_1| T_{K_1}}(s) = \exp\left(-\left[\int_{T_{K_1}}^{0} \frac{P_1 G t}{1+ P_1 G t}\tilde \Lambda^{(1)}({\rm d }t)+  \int_{T_{K_1}}^{0} \frac{P_1 g t}{1+P_1 g t} \tilde {\tilde \Lambda}^{(1)}({\rm d} t)  \right] \right).
\end{align} 
Marginalizing \eqref{eq:rate_condition} regarding $T_1$ and $T_{K_1}$ completes the proof. 
\end{proof}

We note that although Theorem \ref{theo:rcov} relies on Rayleigh fading, the key step to characterize the interference statistic applying BS coordination can be used in a general fading scenario. We verify Theorem \ref{theo:rcov} in the next section.

\section{Numerical Results and Discussions}
In this section, we first verify correctness of the obtained analytical expressions by comparing to simulation results. In addition, we provide useful intuition regarding system design in spectrum-shared mmWave cellular networks.  

\subsection{Inter-Operator BS Coordination}

\begin{figure}[!t]
\centerline{\resizebox{0.65\columnwidth}{!}{\includegraphics{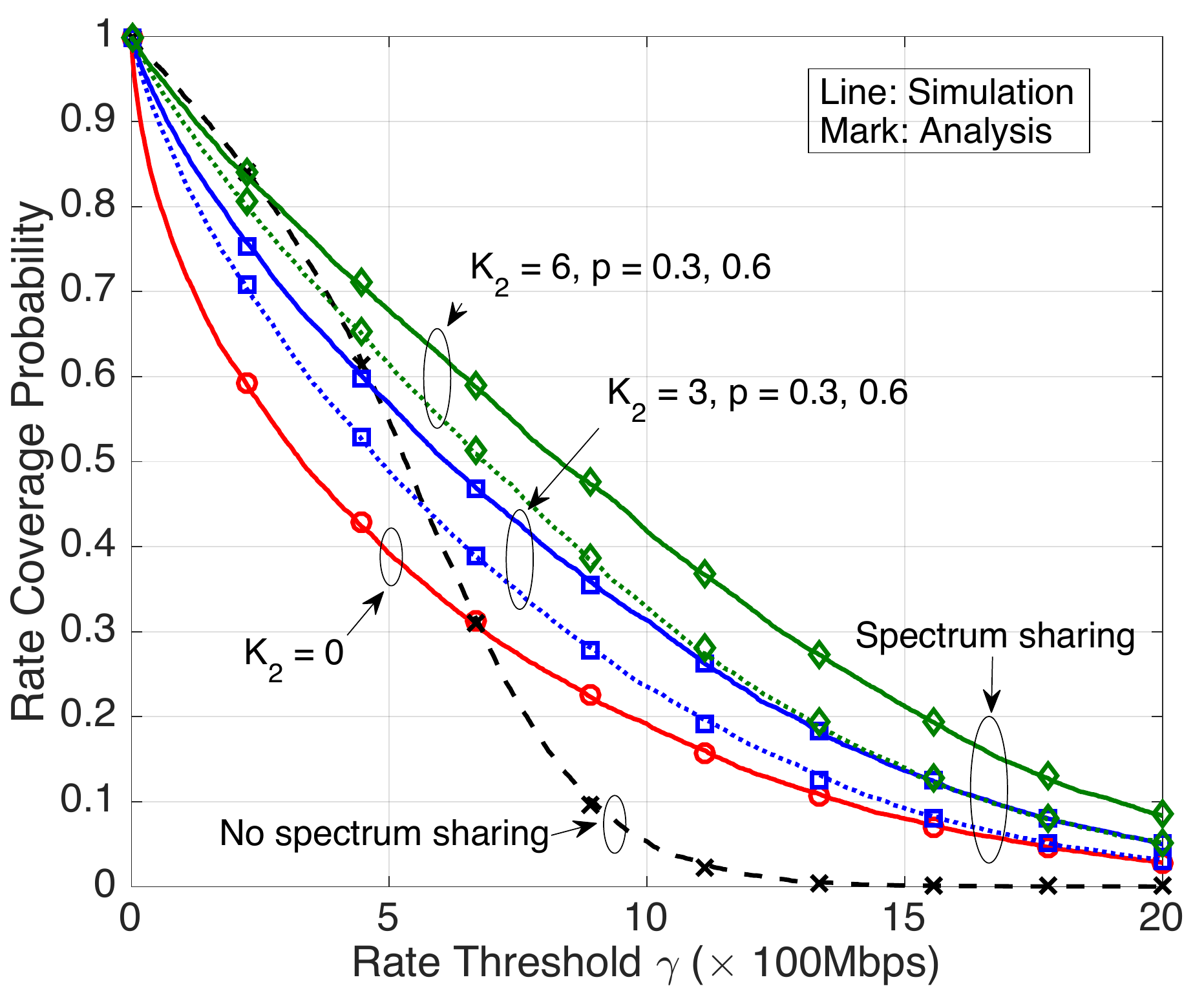}}}     
\caption{The rate coverage probability comparison. We assume that two operator cellular networks share spectrum, i.e, $\left| \CMcal{M}\right| = 2$. The used system parameters are as follows: $\alpha_{\rm L} = 2$, $\alpha_{\rm N} = 4$, $C_{\rm L} = -60{\rm dB}$, $C_{\rm N} = -70 {\rm dB}$, $\mu = 144$, $N = 12$ (beamwidth $30{\degree}$), $p = \{0.3,0.6\}$, $\epsilon = -10{\rm dB}$, and $N_0 = -174{\rm dBm}/{\rm Hz}$. The parameters of each operator is: $P_1 = 20 {\rm dBm}$, $\lambda_1 = 5\times 10^{-5}$, $W_1 = 100 \; {\rm MHz}$, $P_2 = 25 {\rm dBm}$, $\lambda_2 = 10^{-4}$, and $W_2 = 200 \;{\rm MHz}$. }
 \label{fig:interBS}
\end{figure}

We assume that operator $1$ and $2$ share the spectrum, i.e., $\CMcal{M} = \{1,2\}$. The BSs of operator $2$ are coordinated to mitigate the interference to the typical user, while there is no coordination in operator $1$. We write this case as $K_1 = 1$ and $K_2 \ge 0$. We note that although we depict a two operator case for simplicity, our analysis is not limited to this scenario. 
We depict the rate coverage curves in Fig.~\ref{fig:interBS}, where the assumed system parameters are in the caption. For the no spectrum sharing case, we assume $\CMcal{M} = \{1\}$. In Fig.~\ref{fig:interBS}, we assumed that the intra-cluster interference is completely removed and the typical user obtains the fixed beamforming gain $pG$, where $p$ is shown in Fig.~\ref{fig:interBS}. We note that this is the same assumption used in the analysis. 
%In Fig.~\ref{fig:interBS}, we first show that the derived analytical expressions are perfectly matched with the simulation results, which guarantees reliability of our analysis. 
%There are several interesting observations in Fig.~\ref{fig:interBS}

From Fig.~\ref{fig:interBS}, if no BS coordination is used, i.e., $K_2 = 0$, then the spectrum sharing rather decreases the edge and the median rate compared to the no spectrum sharing case. 
This is because operator $2$ has a large density and transmit power compared to those of operator $1$. Therefore a large amount of inter-operator interference impacts the typical user by spectrum sharing. For these reasons, the benefits of using large bandwidth vanish. Directional beamforming may reduce the interference, but the beamwidth is not narrow enough to compensate the performance degradation. 
When applying BS coordination, the typical user gains improve performance with spectrum sharing. For example, the typical user obtains $15\%$ median rate gain when $K_2 = 3$ and $p = 0.6$ and $57\%$ median rate gain when $K_2 = 6$ and $p = 0.6$ over the no spectrum sharing case. In particular, if $K_2 = 6$ and $p = 0.6$, the median rate gain is brought without edge rate degradation. 
Apparently, these gains come from the fact that the BS coordination removes the strongest interference of operator $2$, so that the typical user obtains benefits of using large bandwidth without huge interference. 
We also observe that how $p$ affects the rate coverage performance, where $0 < p \le 1$ is the relative beamforming gain of the typical user. 
 When $p = 0.3$, the significant amount of beamforming gain is lost due to the interference mitigation, so that the inter-operator BS coordination does not increase the rate coverage efficiently. 
%This is especially important when sharing the spectrum with an operator whose density and transmit power is large. 
%Since there is considerable interference of the other operator when the corresponding near the typical user, we expect that the performance gain from BS coordination would become more dramatic as operator $2$ becomes denser. 

We note that inter-operator BS coordination becomes more useful when sharing the spectrum with a dense operator. In \cite{gupta:tcom:16}, it was shown that spectrum sharing offers meaningful performance gain without BS coordination, provided that two operators' densities and transmit power are equal. 
%the probability that the interference power is larger than the desired link increases, which is a main source of rate degradation. Thus, removing those interference 
%This will be demonstrated in the next subsection.
A main source of performance degradation in spectrum sharing is the interference whose power is greater than the desired link. For example, a single interference signal can lead the operating SINR regime below $0 {\rm dB}$ if the interference power is larger than the desired link. 
When the sharing operator is dense, there exist multiple strong interference signals so that interference management via BS coordination is required to extract performance gain. 
On the contrary, when two operators' densities are equal, there are not many strong interference signals, meaning that BS coordination is less necessary.
%only removing them brings significant gains, as observed in Fig.~\ref{fig:interBS}. 
%For this reason, we predict that inter-operator BS coordination would offer more gains as the other operator becomes dense. 
%Since the amount of interference increases as operator $2$ becomes dense, the more dramatic gain would be provided from BS coordination as the larger the density 
From this observation, we conclude that inter-operator BS coordination is a key enabler for sharing the spectrum with a dense operator. 

\subsection{Intra-Operator BS Coordination}

\begin{figure}[!t]
\centerline{\resizebox{0.65\columnwidth}{!}{\includegraphics{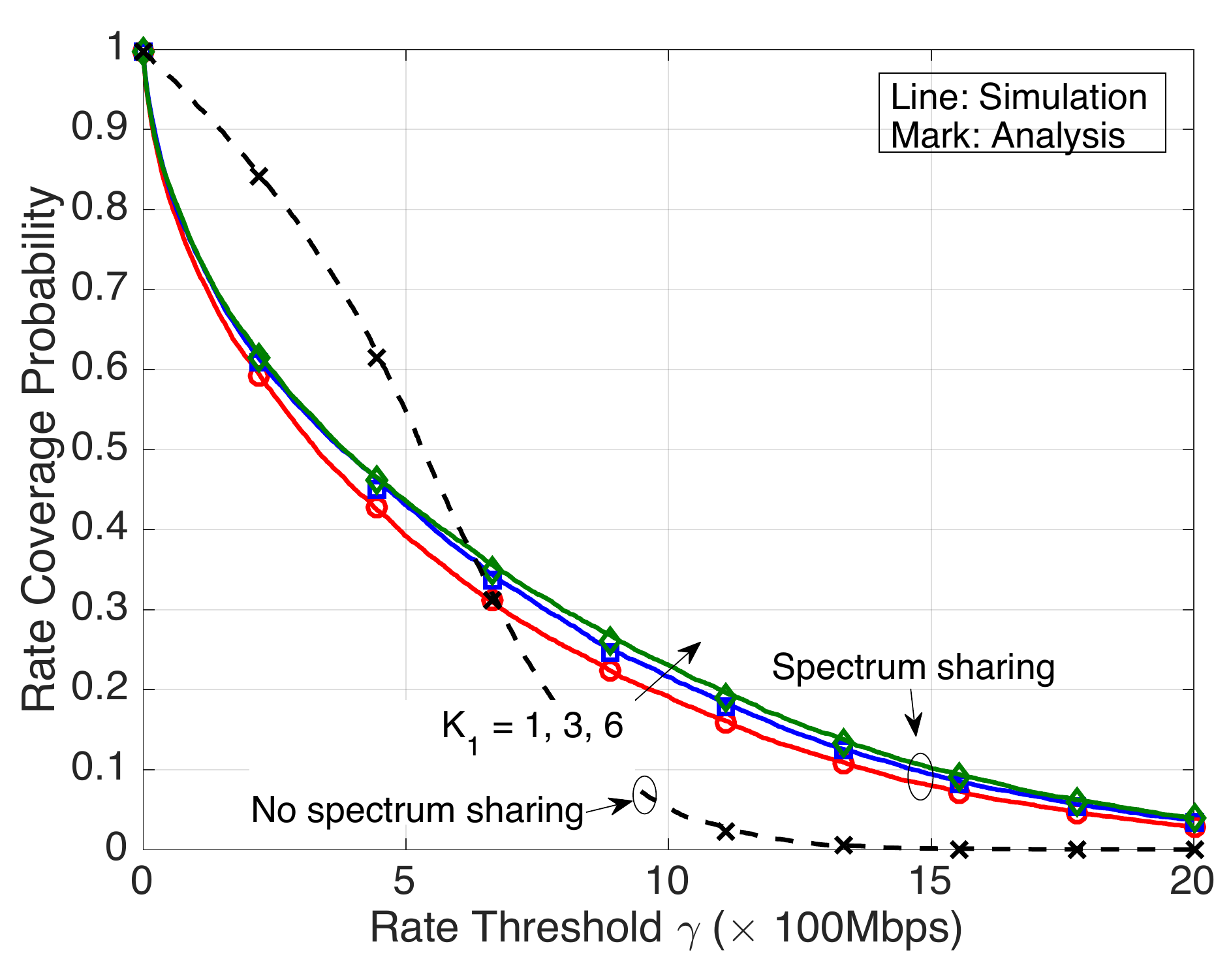}}}     
\caption{The rate coverage probability comparison. We assume two operator cellular networks, i.e, $\left| \CMcal{M}\right| = 2$. The used system parameters are as follows: $\alpha_{\rm L} = 2$, $\alpha_{\rm N} = 4$, $C_{\rm L} = -60{\rm dB}$, $C_{\rm N} = -70 {\rm dB}$, $\mu = 144$, $N = 12$ (beamwidth $30{\degree}$), $p = 1$, $\epsilon = -10{\rm dB}$, and $N_0 = -174{\rm dBm}/{\rm Hz}$. The parameters of each operator is: $P_1 = 20 {\rm dBm}$, $\lambda_1 = 5\times 10^{-5}$, $W_1 = 100\; {\rm MHz}$, $P_2 = 25 {\rm dBm}$, $\lambda_2 = 10^{-4}$, and $W_2 = 200 \;{\rm MHz}$. }
 \label{fig:intraBS}
\end{figure}

We now investigate the performance of intra-operator BS coordination, where BS coordination is applied only in operator $1$. Accordingly, we set $K_2 = 0$. We assume the same system setting in the previous subsection. We illustrate the rate coverage probability in Fig.~\ref{fig:intraBS}, in which we observe that the obtained analytical results are matched with the simulation results. 

Unlike the inter-operator BS coordination, the intra-operator BS coordination does not provide meaningful gain in Fig.~\ref{fig:intraBS}. Specifically, the edge and the median rate are lower than the no spectrum sharing case. In Fig.~\ref{fig:intraBS}, we assume $p = 1$, which is an ideal case that the typical user obtains the beamforming gain without any loss. Since a practical value of $p$ is lower than $1$, the intra-cluster BS provides less rate performance than Fig.~\ref{fig:intraBS} in practice. 
The main rationale is as follows. Using the intra-operator BS coordination, we only remove the interference whose power is weaker than the desired link since the desired link is the strongest link in operator $1$ by the association rule. Unfortunately, this is not much effective since the removed interference is not a main performance impairment factor in the spectrum sharing. A main factor is the interference whose power is larger than the desired link, and the inter-operator BS coordination is able to remove this. Accordingly, when sharing the spectrum with a dense operator, using inter-operator BS coordination is more desirable than intra-operator BS coordination.

Intra-operator BS coordination may not be needed in mmWave cellular network cases, even in the absence of spectrum sharing. Comparing $K_1 = 1$ and $K_2 = 6$, almost no gain is presented. This implies that the out-of-cell interference of operator $1$ is negligible, so removing it does not bring significant performance gain. This is different from the results in \cite{lee:twc:15}, which showed that observable performance gain are obtained by applying dynamic BS coordination in a sub-6 GHz cellular network. The difference comes from that, in a mmWave network, directional beamforming and vulnerability to blockages of mmWave signals inherently reduces the interference. For this reason, we claim that there is no need to make an effort to mitigate the out-of-cell interference in a mmWave cellular network. The prior work \cite{sarabjot:jsac:15} backs this claim. It showed that, in a single-tier mmWave cellular network, the SINR operates in a noise-limited regime since the out-of-cell interference is trivial.  

%the interference whose power is smaller than the desired link is not much a problem in a mmWave cellular network. 
%The observation in Fig.~\ref{fig:intraBS} also gives intuition regarding BS coordination in a single-tier mmWave cellular network. 

%This implies that the out-of-cell interference of the operator $1$ is not that severe and it is not an important factor in determination of the rate performance. This is because directional beamforming and vulnerability to blockages of mmWave signals inherently reduces the interference. Based on this intuition, there is no need to make an effort to mitigate the out-of-cell interference. The prior result \cite{sarabjot:jsac:15} backs our claim. It showed that, in a single-tier mmWave cellular network, the SINR operates in a noise-limited regime since the out-of-cell interference is trivial.  
%Since interference management via BS coordination is expected to have only marginal gain in a noise-limited regime, BS coordination is not required in a single mmWave cellular network. 

\subsection{The Beamwidth Effect}

\begin{figure}[!t]
\centerline{\resizebox{0.65\columnwidth}{!}{\includegraphics{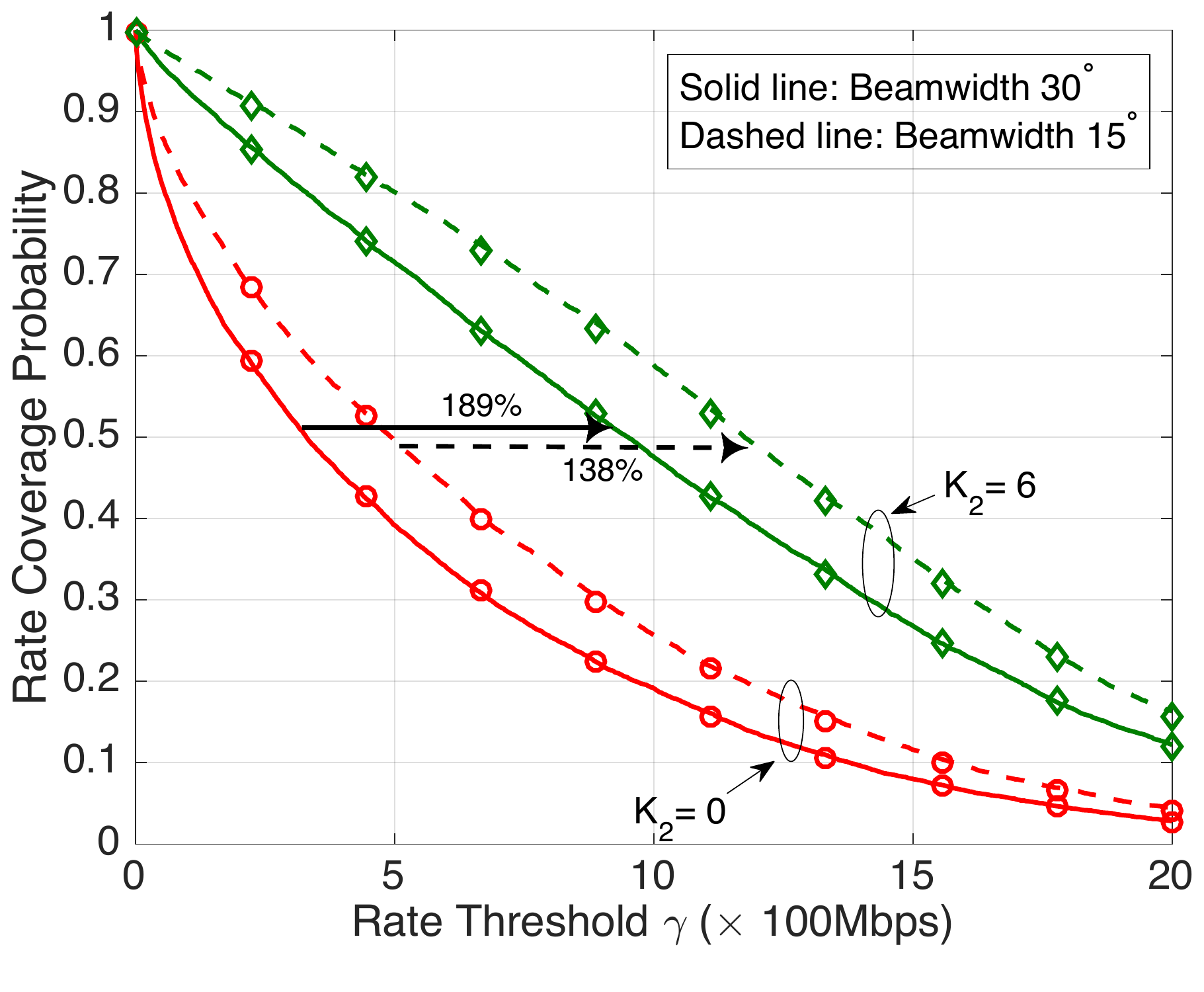}}}     
\caption{The rate coverage probability comparison. We assume two operator cellular networks, i.e, $\left| \CMcal{M}\right| = 2$. The used system parameters are as follows: $\alpha_{\rm L} = 2$, $\alpha_{\rm N} = 4$, $C_{\rm L} = -60{\rm dB}$, $C_{\rm N} = -70 {\rm dB}$, $\mu = 144$, $N = \{12,24\}$ (beamwidth $\{30{\degree}, 15{\degree}\}$), $p = 1$, $\epsilon = -10{\rm dB}$, and $N_0 = -174{\rm dBm}/{\rm Hz}$. The parameters of each operator is: $P_1 = 20 {\rm dBm}$, $\lambda_1 = 5\times 10^{-5}$, $W_1 = 100 \;{\rm MHz}$, $P_2 = 25 {\rm dBm}$, $\lambda_2 = 10^{-4}$, and $W_2 = 200\; {\rm MHz}$. }
 \label{fig:width_compare}
\end{figure}

We explore how the BS coordination gain is changed depending on the beamwidth. 
We draw the rate coverage graphs in Fig.~\ref{fig:width_compare}, assuming the beamwidth $15\degree$ and $30\degree$, and the inter-operator BS coordination with $K_2 = 0$ and $K_2 = 6$. In Fig.~\ref{fig:width_compare}, it is shown that the narrow beam case has higher rate coverage performance than the wide beam case. This is not surprising because the amount of interference decreases as the beamwidth becomes smaller. 
Now we examine the relative median rate gain of the BS coordination. As described in Fig.~\ref{fig:width_compare}, the wide beam case has $189\%$ median rate gain with $K_2 = 6$, while the narrow beam case has $138\%$ median rate gain in the same coordination environment. This means that the BS coordination is more efficient in the wide beam case than in the narrow beam case. 
%We find two reasons to explain this observation. First, the BS coordination offers higher performance gain when the SINR is low. This is because the interference is the denominator of the SINR, so reducing the interference drastically increases the SINR when it has a low value. Since the wide beam case has lower SINR, larger gains are presented.
We explain the reason as follows. The coordination set includes the BSs based on their link power $C_{s} \left\|{\bf{d}}_i \right\|^{-\alpha_s}$. The actual interference, however, also incorporates the directionality gain $G\left( \theta_i\right)$. For this reason, there is a non-zero probability that a BS not included in a coordination set incurs strong interference to the typical user. In such a case, BS coordination does not efficiently cancel the significant interference, leading to unsatisfactory gain. 
Since the directionality gain increases as a beam becomes narrow, this probability increases in the narrow beam case, providing less gain than in the wide beam case. 

In a system design perspective, using narrow beams and BS coordination are complementary each other in a role of interference management. For instance, if we cannot use narrow beams, BS coordination can be applied alternatively to mitigate the interference. In this case, we expect considerable gain from BS coordination as it is more efficient with a wide beam. On the contrary, if narrow beams are available to use, BS coordination is not necessary since the interference is already reduced by directionality. 

We also conjecture that BS coordination is efficient in rich scattering. When there is sufficient scattering, signals arrive to a user from many directions, so that directional beamforming gain decreases. As an extreme case, when there is very rich scattering and dispersion, this becomes equivalent when no directional beamforming is used. In this case, the interference signal power is determined only by the corresponding path-loss, while the directionality gain becomes negligible. For this reason, the BSs whose strong interference signal power are included in a coordination set, so that the interference is efficiently managed. For this reason, we expect that the inter-operator BS coordination is useful as there is a lot of scattering. 

%BS coordination is more efficient in the wide beam case because of the above reasons. 
%using narrow beams and BS coordination both cause overheads, thereby more narrow beam can be used without BS coordination, while BS coordination 

\subsection{LoS/NLoS Ratio in a Coordination Set}
\begin{figure}[!t]
\centerline{\resizebox{0.65\columnwidth}{!}{\includegraphics{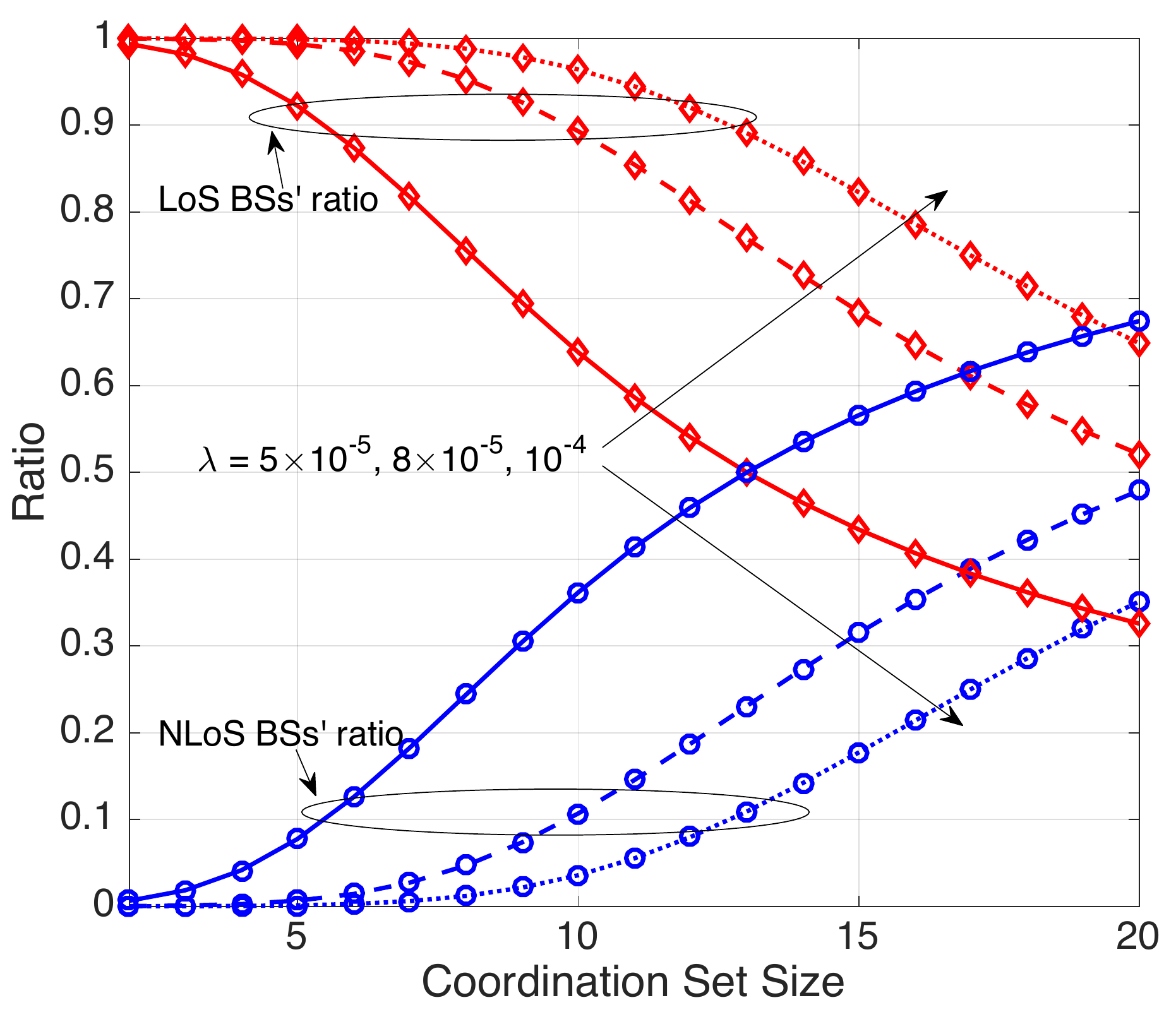}}}     
\caption{The ratio of LoS and NLoS BSs in a coordination set. The assumed system parameters are as follows: $\alpha_{\rm L} = 2$, $\alpha_{\rm N} = 4$, $C_{\rm L} = -60{\rm dB}$, $C_{\rm N} = -70 {\rm dB}$, $\lambda = \{5 \times 10^{-5}, 8 \times 10^{-5}, 10^{-4}\}$, $\left| \CMcal{A}\right| = \{1,..,20\}$, and $\mu = 144$.}
 \label{fig:LoS_ratio}
\end{figure}

In this subsection, we study the ratio of LoS and NLoS BSs in a coordination set. The average ratio is analytically derived in Example \ref{example:LoSratio}. We plot the obtained analytical expression in Fig.~\ref{fig:LoS_ratio} assuming various densities and coordination set size. 
As observed in Fig.~\ref{fig:LoS_ratio}, the ratio of LoS BSs increases as (i) the coordination set size decreases and (ii) the density increases. 
Basically, a LoS BS has large power compared to a NLoS BS. Therefore, LoS BSs have higher probability to be included in a coordination set compared to NLoS BSs, if there distances are comparable. Nevertheless, since the number of LoS BSs exponentially decreases as their distances increase, the existing LoS BSs' distances exponentially increase. 
When the coordination set size is small, only BSs located near the origin are candidates to be included in a coordination set. In this case, LoS BSs and NLoS BSs' distances are similar, therefore the coordination set consists of mostly LoS BSs. As the set size increases, however, the coordination set is likely to include further located BSs. The existing LoS BSs' distances are much larger than NLoS BSs, the LoS BSs' ratio decreases. 
When the density increases, there are more LoS BSs within similar distances to NLoS BSs, so the LoS BSs' ratio increases. 
This result provides insights regarding the states of strong interference in mmWave cellular networks. For example, if $\lambda = 8 \times 10^{-5}$, about $90\%$ of the $10$ strong interference is LoS, while this decreases to $65\%$ if $\lambda = 5 \times 10^{-5}$. For this reason, assuming that only LoS interference can be cancelled, such interference cancellation may be efficient when $\lambda = 8 \times 10^{-5}$ since $90\%$ of strong interference can be removed. If $\lambda = 5 \times 10^{-5}$, however, this can remove only $65\%$ of strong interference and there is remaining $45\%$ NLoS interference that decreases the SINR. For this reason, the assumed interference cancellation may not be efficient when $\lambda = 5 \times 10^{-5}$ compared to $\lambda = 8 \times 10^{-5}$. 
%This ratio comparison give some insights for forming an efficient coordination set. 
%When the density increases, there are more LoS BSs in the same distance, so its ratio increases. 
%As the evidence of our interpretation, Fig.~\ref{fig:LoS_ratio} shows that the LoS BSs' ratio increases as the density increases. When the density is large, there are many LoS BSs and these LoS BSs are likely to be included in $\CMcal{A}$ since their link power is significantly stronger than that of NLoS BSs. 

\section{Conclusions}
We analyzed the rate coverage probability of a spectrum-shared mmWave cellular network employing inter-operator BS coordination. 
As a key step, we derived the PDF of the $K$-th strongest link power incorporating the LoS and NLoS states and showed the the obtained PDF reduces to previous results in simpler special cases. 
Leveraging this PDF, we obtained the rate coverage expression as a function of key system parameters, chiefly the coordination set size. 
In the simulation, we verified the correctness of the analysis.
Major findings from the analysis are as follows. First, inter-operator BS coordination provides significant performance gains in spectrum sharing with a dense operator, where the gains comes from mitigating the inter-operator interference. 
%which is significant particularly when the other operator is dense or has high transmit power. 
Second, intra-operator BS coordination does not leas to much performance improvement. 
%Due to the directionality, the intra-operator interference already decreases enough.
Third, the inter-operator BS coordination is more efficient when there are wide beams. In summary, we show that the inter-operator BS coordination is necessary in a spectrum-shared mmWave cellular system, especially when sharing the spectrum with a dense and high power operator.

There are several directions for future work. 
One can consider the overheads associated with measuring CSIT and beam direction. Considering this, the beamwidth and the coordination set size can be jointly optimized. For example, assuming that finding an exact beam direction is difficult, a large coordination set might be preferred. 
Beyond our simple ZF-like hybrid precoder, one can also consider more advanced precoding methods as in \cite{spark:tsp:17}. 
%\cite{sohrabi:jstsp:16, spark:tsp:17}. 
%Extending advanced hybrid precoding methods to BS coordination for managing the interference is interesting.
%Exploring the performance of the BS coordination with this type of precoding is interesting. 
In addition, considering a partial loading effect \cite{gupta:tcom:16} is also promising. If the user density is less than the BS density in a particular operator, some BSs can be turned off since there is no user associated with them. Under this assumption, the performance of inter-operator BS coordination may be affected not only by the BS density, but also the user density. 
Next, one can use a more realistic channel model beyond a simple sectored antenna model. As shown in \cite{valenzuela:ctw:17}, a sectored antenna model can be far from real mmWave networks, especially for NLoS signals because of scattering and dispersion. Since our key lemma regarding the $K$-th strongest BS's link power (Lemma \ref{lem:PDF_K}) does not depend on a specific channel model, one can use our results to characterize the performance by using realistic channel models.
Finally, by using the developed analytical framework, one can consider more sophisticated sharing scenario, for example sharing with a Wi-Fi  \cite{li:twc:16}, satellite service \cite{guidolin:gc:15}, and radar \cite{raymond:jsac:16}, or access and infrastructure sharing with other operators \cite{kibilda:twc:17}. 

\appendices
\section{Proof of Lemma 1} \label{appen:lem1}
We first present that the CDF of the $K$-th strongest BS's link power, i.e., $F_{T_K}(t)$, is equivalent to the probability that there are less than $K$ BSs (including LoS and NLoS BSs) whose link power is larger than $t$. 
To describe this probability, we define $N_{\rm L}(t_1, t_2)$ (or $N_{\rm N}(t_1, t_2)$) as the number of LoS (or NLoS) BSs whose link power is in the region $\left(t_1, t_2 \right]$. Than, the CDF $F_{T_K}(t)$ is obtained as follows
\begin{align}
F_{T_K}(t)  =
\mathbb{P}\left[T_K < t \right]  
&= \sum_{k = 0}^{K-1} \sum_{\ell = 0}^{K-k-1} \mathbb{P}\left[\{ N_{\rm L}(t, \infty) = k\} \cap \{N_{\rm N}(t, \infty) = \ell\} \right] \nonumber \\
& \mathop{=}^{(a)} \sum_{k = 0}^{K-1} \sum_{\ell = 0}^{K-k-1} \mathbb{P}\left[ N_{\rm L}(t, \infty) = k\right] \mathbb{P}\left[ N_{\rm N}(t, \infty) = \ell \right] \nonumber \\
& \mathop{=}^{(b)} e^{-\Lambda_{\rm L}(t)} e^{-\Lambda_{\rm N}(t) }\sum_{k = 0}^{K-1} \sum_{\ell = 0}^{K-k-1} \frac{\Lambda_{\rm L}(t) ^ k }{k !} \frac{\Lambda_{\rm N}(t) ^{\ell} }{\ell !},
\end{align}
where (a) follows that $\Phi_{\rm L}$ and $\Phi_{\rm N}$ are mutually independent and (b) follows Displacement theorem \cite{baccelli:inria} where the number of the BSs in a closed region is a Poisson random variable with the intensity measure $\Lambda_{\rm L}(t)$ and $\Lambda_{\rm N}(t)$. The intensity measure $\Lambda_{\rm L}(t)$ is computed as
\begin{align}
\Lambda_{\rm L}(t) =& \mathbb{E}\left[\sum_{{\bf{d}} \in \Phi_{\rm L}} {\bf{1}}\left(C_{\rm L} \left\| {\bf{d}} \right\|^{-\alpha_{\rm L}}>t \right) \right] 
\mathop{=}^{(a)}  2\pi \lambda \int_{0}^{\psi_{\rm L}(t)} p(x) x {\rm d}x, 
\end{align}
where (a) follows Campbell's theorem with $\psi_{\rm L}(t) = \left(\frac{C_{\rm L}}{t} \right)^{\frac{1}{\alpha_{\rm L}}}$. Similar to this, we have
\begin{align}
\Lambda_{\rm N}(t) 
%=& \mathbb{E}\left[\sum_{{\bf{d}} \in \Phi_{\rm N}} {\bf{1}}\left(C_{\rm N} \left\| {\bf{d}} \right\|^{-\alpha_{\rm N}}>t \right) \right] \nonumber \\
=&  2\pi \lambda \int_{0}^{\psi_{\rm N}(r)} (1-p(x)) x {\rm d}x, 
\end{align}
where $\psi_{\rm N}(r) = \left(\frac{C_{\rm N}}{t} \right)^{\frac{1}{\alpha_{\rm N}}}$. Subsequently, we obtain the PDF of $T_K$ by differentiating $F_{T_K}(t)$ as follows 
\begin{align}
f_{T_K}(t)  = \frac{\partial F_{T_K}(t)}{\partial t} 
=& -\Lambda'_{\rm L}(t) e^{-\Lambda_{\rm L}(t) } e^{-\Lambda_{\rm N}(t)} \sum_{k = 0}^{K-1} \sum_{\ell = 0}^{K-k-1} \frac{\Lambda_{\rm L}(t) ^ k }{k !} \frac{\Lambda_{\rm N}(t) ^{\ell} }{\ell !} \nonumber \\
& +\Lambda'_{\rm L}(t) e^{-\Lambda_{\rm L}(t) } e^{-\Lambda_{\rm N}(t)} \sum_{k = 1}^{K-1} \sum_{\ell = 0}^{K-k-1} \frac{\Lambda_{\rm L}(t) ^{k-1} }{(k-1) !} \frac{\Lambda_{\rm N}(t) ^{\ell} }{\ell !} \nonumber \\
& -\Lambda'_{\rm N}(t) e^{-\Lambda_{\rm L}(t) } e^{-\Lambda_{\rm N}(t)} \sum_{k = 0}^{K-1} \sum_{\ell = 0}^{K-k-1} \frac{\Lambda_{\rm L}(t) ^ k }{k !} \frac{\Lambda_{\rm N}(t) ^{\ell} }{\ell !} \nonumber \\
& +\Lambda'_{\rm N}(t) e^{-\Lambda_{\rm L}(t) } e^{-\Lambda_{\rm N}(t)} \sum_{k = 0}^{K-1} \sum_{\ell = 1}^{K-k-1} \frac{\Lambda_{\rm L}(t) ^ k }{k !} \frac{\Lambda_{\rm N}(t) ^{\ell-1} }{(\ell-1) !}.
\end{align}
Calculating the first $\&$ the second term, and the third term $\&$ the fourth term separately, we have
\begin{align}
f_{T_K}(t) =& -\Lambda_{\rm L}'(t) e^{-\Lambda_{\rm L}(t) } e^{-\Lambda_{\rm N}(t)} \left(\sum_{k=0}^{K-1}\frac{\Lambda_{\rm L}(t)^k}{k!} \frac{\Lambda_{\rm N}(t)^{K-k-1}}{(K-k-1)!} \right) \nonumber \\
& -\Lambda_{\rm N}'(t) e^{-\Lambda_{\rm L}(t) } e^{-\Lambda_{\rm N}(t)} \left(\sum_{k=0}^{K-1}\frac{\Lambda_{\rm L}(t)^k}{k!} \frac{\Lambda_{\rm N}(t)^{K-k-1}}{(K-k-1)!} \right).
\end{align}
Combining the two terms,
\begin{align}
f_{T_K}(t) 
=& e^{-\Lambda_{\rm L}(t) } e^{-\Lambda_{\rm N}(t)} \left(\sum_{k=0}^{K-1}\frac{\Lambda_{\rm L}(t)^k}{k!} \frac{\Lambda_{\rm N}(t)^{K-k-1}}{(K-k-1)!} \right) \cdot \left(-\Lambda_{\rm L}'(t) - \Lambda_{\rm N}'(t) \right) \nonumber \\
=& e^{-\Lambda_{\rm L}(t) } e^{-\Lambda_{\rm N}(t)} \frac{\left(\Lambda_{\rm L}(t) + \Lambda_{\rm N}(t) \right)^{K-1}}{(K-1)!}\cdot \left(-\Lambda_{\rm L}'(t) - \Lambda_{\rm N}'(t) \right),
\end{align}
which completes the proof.

\section{Proof of Lemma 2} \label{appen:lem2}
We start with characterizing the joint CDF of $T_1$ and $T_K$ denoted as $F_{T_1, T_K}(t_1, t_K)$ assuming that $t_K < t_1$. The joint CDF $F_{T_1, T_K}(t_1, t_K)$ is the probability that there are $1\le k < K$ BS whose link power is in $\left[t_K, t_1\right)$, and simultaneously, there is no BS whose link power is in $\left[t_1, \infty\right)$. 
\begin{align}
& F_{T_1, T_K}(t_1, t_K) = \mathbb{P}\left[\{T_1 < t_1\} \cap \{T_K < t_K\} \right] \nonumber \\
&= \sum_{k = 1}^{K-1} \sum_{\ell = 0}^{k} \mathbb{P}\left[\{N_{\rm L}(t_K, t_1) = \ell\} \cap \{N_{\rm N}(t_K, t_1) = k-\ell \}\right] \cdot \mathbb{P}\left[\{N_{\rm L}\left(t_1, \infty \right) =0\}\cap \{N_{\rm N}(t_1, \infty) = 0\}\right] \nonumber \\
%&= e^{-\Lambda_{\rm L}(t_K, t_1)} e^{-\Lambda_{\rm L}(t_1)}  e^{-\Lambda_{\rm N}(t_K, t_1)}  e^{-\Lambda_{\rm N}(t_1)}\sum_{k = 1}^{K-1} \sum_{\ell = 0}^{k} \frac{\Lambda_{\rm L}(t_K, t_1)^{\ell}}{\ell!} \frac{\Lambda_{\rm N}(t_K, t_1)^{k-\ell}}{(k-\ell)!} \nonumber \\
&= e^{-\Lambda_{\rm L}(t_K)} e^{-\Lambda_{\rm N}(t_K)} \sum_{k = 1}^{K-1} \sum_{\ell = 0}^{k} \frac{\Lambda_{\rm L}(t_K, t_1)^{\ell}}{\ell!} \frac{\Lambda_{\rm N}(t_K, t_1)^{k-\ell}}{(k-\ell)!}
\end{align}
where $\Lambda_{\rm L}(t_K, t_1)$ is the differential intensity measure defined as
\begin{align}
\Lambda_{\rm L}(t_K, t_1) =& \mathbb{E}\left[\sum_{{\bf{d}} \in \Phi_{\rm L}} {\bf{1}}\left(t_K < C_{\rm L} \left\| {\bf{d}} \right\|^{-\alpha_{\rm L}} < t_1 \right) \right] \nonumber \\
\mathop{=}^{(a)}&  2\pi \lambda \int_{\psi_{\rm L}(t_1)}^{\psi_{\rm L}(t_K)} p(x)  x {\rm d} x 
\mathop{=}^{} \Lambda_{\rm L}(t_K) - \Lambda_{\rm L}(t_1),
\end{align}
where (a) follows Campbell's theorem. The joint PDF is obtained by differentiating $F_{T_1, T_K}(t_1, t_K)$ by $t_1$ and $t_K$, 
\begin{align}
f_{T_1, T_K}(t_1, t_K) =& \frac{\partial^2 F_{T_1, T_K}(t_1 t_K)}{\partial t_1 \partial t_K}.
\end{align}
We first compute the derivative regarding $t_K$. 
\begin{align} \label{eq:jointPDF_tk}
\frac{\partial F_{T_1, T_K(t_1, t_K)}}{\partial t_K} =& -\Lambda'_{\rm L}(t_K) e^{-\Lambda_{\rm L}(t_K)} e^{-\Lambda_{\rm N}(t_K)} \sum_{k = 1}^{K-1} \sum_{\ell = 0}^{k} \frac{\Lambda_{\rm L}(t_K, t_1)^{\ell}}{\ell!} \frac{\Lambda_{\rm N}(t_K, t_1)^{k-\ell}}{(k-\ell)!} \nonumber \\
& +\Lambda'_{\rm L}(t_K) e^{-\Lambda_{\rm L}(t_K)} e^{-\Lambda_{\rm N}(t_K)} \sum_{k = 1}^{K-1} \sum_{\ell = 1}^{k} \frac{\Lambda_{\rm L}(t_K, t_1)^{(\ell-1)}}{(\ell-1)!} \frac{\Lambda_{\rm N}(t_K, t_1)^{k-\ell}}{(k-\ell)!} \nonumber \\
& -\Lambda'_{\rm N}(t_K) e^{-\Lambda_{\rm L}(t_K)} e^{-\Lambda_{\rm N}(t_K)} \sum_{k = 1}^{K-1} \sum_{\ell = 0}^{k} \frac{\Lambda_{\rm L}(t_K, t_1)^{\ell}}{\ell!} \frac{\Lambda_{\rm N}(t_K, t_1)^{k-\ell}}{(k-\ell)!}  \nonumber \\
& + \Lambda'_{\rm N}(t_K) e^{-\Lambda_{\rm L}(t_K)} e^{-\Lambda_{\rm N}(t_K)} \sum_{k = 1}^{K-1} \sum_{\ell = 0}^{k-1} \frac{\Lambda_{\rm L}(t_K, t_1)^{\ell}}{\ell!} \frac{\Lambda_{\rm N}(t_K, t_1)^{k-\ell-1}}{(k-\ell-1)!}.
\end{align}
Simplifying \eqref{eq:jointPDF_tk}, we have
\begin{align} \label{eq:jointPDF_tk_simple}
\frac{\partial F_{T_1, T_K(t_1, t_K)}}{\partial t_K} =& -\Lambda_{\rm L}'(t_K) e^{-\Lambda_{\rm L}(t_K)} e^{-\Lambda_{\rm N}(t_K)} \left(\sum_{k = 0}^{K-1}\frac{\Lambda_{\rm L}(t_K, t_1)^{k}}{k!} \frac{\Lambda_{\rm N}(t_K, t_1)^{K-1-k}}{(K-1-k)!} -1\right) \nonumber \\
& -\Lambda_{\rm N}'(t_K) e^{-\Lambda_{\rm L}(t_K)} e^{-\Lambda_{\rm N}(t_K)} \left(\sum_{k = 0}^{K-1} \frac{\Lambda_{\rm L}(t_K, t_1)^{k}}{k!} \frac{\Lambda_{\rm N}(t_K, t_1)^{K-1-k}}{(K-1-k)!} -1\right).
\end{align}
Next, we differentiate \eqref{eq:jointPDF_tk_simple} regarding $t_1$. 
\begin{align}
\frac{\partial^2 F_{T_1, T_K(t_1, t_K)}}{\partial t_K \partial t_1} =& \Lambda'_{\rm L}(t_K) \Lambda'_{\rm L}(t_1) e^{-\Lambda_{\rm L}(t_K)} e^{-\Lambda_{\rm N}(t_K)} \left(\sum_{k = 1}^{K-1}\frac{\Lambda_{\rm L}(t_K, t_1)^{k-1}}{(k-1)!} \frac{\Lambda_{\rm N}(t_K, t_1)^{K-1-k}}{(K-1-k)!} \right) \nonumber \\
& + \Lambda'_{\rm L}(t_K) \Lambda'_{\rm N}(t_1) e^{-\Lambda_{\rm L}(t_K)} e^{-\Lambda_{\rm N}(t_K)} \left(\sum_{k = 0}^{K-2}\frac{\Lambda_{\rm L}(t_K, t_1)^{k}}{k!} \frac{\Lambda_{\rm N}(t_K, t_1)^{K-2-k}}{(K-2-k)!} \right) \nonumber \\
& + \Lambda'_{\rm N} (t_K) \Lambda'_{\rm L}(t_1) e^{-\Lambda_{\rm L}(t_K)} e^{-\Lambda_{\rm N}(t_K)} \left(\sum_{k = 1}^{K-1}\frac{\Lambda_{\rm L}(t_K, t_1)^{k-1}}{(k-1)!} \frac{\Lambda_{\rm N}(t_K, t_1)^{K-1-k}}{(K-1-k)!} \right) \nonumber \\
& + \Lambda'_{\rm N} (t_K) \Lambda'_{\rm N}(t_1) e^{-\Lambda_{\rm L}(t_K)} e^{-\Lambda_{\rm N}(t_K)} \left(\sum_{k = 0}^{K-2}\frac{\Lambda_{\rm L}(t_K, t_1)^{k}}{k!} \frac{\Lambda_{\rm N}(t_K, t_1)^{K-2-k}}{(K-2-k)!} \right). 
\end{align}
Further simplifying, we finally have
\begin{align}
& \frac{\partial^2 F_{T_1, T_K(t_1, t_K)}}{\partial t_K \partial t_1} \nonumber \\
%=& \Lambda'_{\rm L}(t_K) \Lambda'_{\rm L}(t_1) e^{-\Lambda_{\rm L}(t_K)} e^{-\Lambda_{\rm N}(t_K)} \left(\frac{\left(\Lambda_{\rm L}(t_K, t_1) + \Lambda_{\rm N}(t_K, t_1)\right)^{(K-2)}}{(K-2)!}  \right) \nonumber \\
%& + \Lambda'_{\rm L}(t_K) \Lambda'_{\rm N}(t_1) e^{-\Lambda_{\rm L}(t_K)} e^{-\Lambda_{\rm N}(t_K)} \left(\frac{\left(\Lambda_{\rm L}(t_K, t_1) + \Lambda_{\rm N}(t_K, t_1)\right)^{(K-2)}}{(K-2)!}  \right) \nonumber \\
%& + \Lambda'_{\rm N} (t_K) \Lambda'_{\rm L}(t_1) e^{-\Lambda_{\rm L}(t_K)} e^{-\Lambda_{\rm N}(t_K)} \left(\frac{\left(\Lambda_{\rm L}(t_K, t_1) + \Lambda_{\rm N}(t_K, t_1)\right)^{(K-2)}}{(K-2)!}  \right) \nonumber \\
%& + \Lambda'_{\rm N} (t_K) \Lambda'_{\rm N}(t_1) e^{-\Lambda_{\rm L}(t_K)} e^{-\Lambda_{\rm N}(t_K)} \left(\frac{\left(\Lambda_{\rm L}(t_K, t_1) + \Lambda_{\rm N}(t_K, t_1)\right)^{(K-2)}}{(K-2)!}  \right) \nonumber \\
=& \left(\begin{array}{c}{ \Lambda'_{\rm L}(t_K) \Lambda'_{\rm L}(t_1) + \Lambda'_{\rm L}(t_K) \Lambda'_{\rm L}(t_1) } \\{+ \Lambda'_{\rm L}(t_K) \Lambda'_{\rm L}(t_1) + \Lambda'_{\rm L}(t_K) \Lambda'_{\rm L}(t_1)} \end{array} \right) \cdot e^{-\Lambda_{\rm L}(t_K)} e^{-\Lambda_{\rm N}(t_K)} \frac{\left(\Lambda_{\rm L}(t_K, t_1) + \Lambda_{\rm N}(t_K, t_1)\right)^{K-2}}{(K-2)!}.
\end{align}
This completes the proof.

\section{Proof of Lemma 3} \label{appen:lem3}
We first transform $C_{s} \left\|{\bf{d}}_i^{(m)} \right\|^{-\alpha_s}$ to the link power variable $t_i$. Conditioning on that the $K_m$-th strongest link power in the coordination set $\CMcal{A}_m$ is $T$, the interference is written as
\begin{align}
I_m = \sum_{t_i < T} P_m G(\theta_{i}^{(m)}) H_i^{(m)} t_i,
\end{align}
where $H_i^{(m)} =\left|\tilde \beta_i^{(m)}\right|^2 \sim {\rm Exp}(1)$. Due to the Displacement theorem, $I_m$ is a one dimensional PPP with the intensity measure $\Lambda^{(m)}(t)$ defined as
\begin{align}
\Lambda^{(m)} (t) = \mathbb{E}\left[\sum_{{\bf{d}} \in \Phi_m} {\bf{1}}\left( C_{s} \left\| {\bf{d}} \right\|^{-\alpha_s} > t \right) \right].
\end{align}
Following the property of a PPP, the conditional Laplace transform of $I_m$ is written as
\begin{align}
&\CMcal{L}_{I_m | T}(s)  = \mathbb{E}\left[\left. e^{-s I_m} \right| T \right] \nonumber \\
&\mathop{=}^{(a)} \exp\left(- \int_{T}^{0} \frac{sP_m G(\theta_{i}^{(m)})t}{1 +sP_m G(\theta_{i}^{(m)})t } \Lambda^{(m)} ({\rm d} t) \right) \nonumber \\
&\mathop{=}^{(b)} \exp\left(- \left[\int_{T}^{0} \frac{sP_m G t}{1 +sP_m Gt } \tilde \Lambda^{(m)} ({\rm d} t) + \int_{T}^{0} \frac{sP_m gt}{1 +sP_m gt } \tilde {\tilde \Lambda}^{(m)} ({\rm d} t) \right]  \right),
\end{align}
where (a) follows the Laplace transform of an exponential random variable with unit mean and also the PGFL of a PPP, and (b) follows independent thinning with $\tilde \Lambda^{(m)}(t) = 1 / N  \cdot \Lambda^{(m)}(t)$ and $\tilde {\tilde \Lambda}^{(m)}(t) = (1 - 1/N) \cdot \Lambda^{(m)}(t)$. We characterize the intensity measure $ \Lambda^{(m)}( t)$ as follows 
\begin{align}
\Lambda^{(m)} (t) =& \mathbb{E}\left[\sum_{{\bf{d}} \in \Phi_{m}} {\bf{1}}\left(C_{s} \left\| {\bf{d}} \right\|^{-\alpha_{s}}>t \right) \right] \nonumber \\
=& \mathbb{E}\left[\sum_{{\bf{d}} \in \Phi_{m,{\rm L}}} {\bf{1}}\left(C_{s} \left\| {\bf{d}} \right\|^{-\alpha_{s}}>t \right) \right] + \mathbb{E}\left[\sum_{{\bf{d}} \in \Phi_{m, {\rm N}}} {\bf{1}}\left(C_{s} \left\| {\bf{d}} \right\|^{-\alpha_{s}}>t \right) \right] \nonumber \\
=& 2\pi \lambda_m \int_{0}^{\left(\frac{C_{\rm L}}{t} \right)^{1/\alpha_{\rm L}}} p(x) x{\rm d} x + 2\pi \lambda_m \int_{0}^{\left(\frac{C_{\rm N}}{t} \right)^{1/\alpha_{\rm N}}} (1-p(x)) x{\rm d} x.
\end{align}
Then we have 
\begin{align}
\Lambda^{(m)}({\rm d} t) =  -\frac{2\pi \lambda_m }{\alpha_{\rm L} C_{\rm L}} p\left(\left(\frac{C_{\rm L}}{t} \right)^{1/\alpha_{\rm L}} \right) \left( \frac{C_{\rm L}}{t} \right)^{2/\alpha_{\rm L}+1} {\rm d} t - \frac{2\pi \lambda_m}{\alpha_{\rm N} C_{\rm N}} \left(1 - p\left(\left(\frac{C_{\rm N}}{t} \right)^{1/\alpha_{\rm N}} \right) \right) \left( \frac{C_{\rm N}}{t}\right)^{2/\alpha_{\rm N}+1} {\rm d} t.
\end{align}
For marginalizing with $T$, we use the PDF $f^{(m)}_{T}(t)$. Finally we reach 
\begin{align}
&\CMcal{L}_{I_m}(s) = \int_{T = 0}^{\infty}  \CMcal{L}_{I_m | T}(s) f^{(m)}_{T_{K_m}}(T) {\rm d} T,
\end{align}
which completes the proof. 

\bibliographystyle{IEEEtran}
\bibliography{ref_mcc_mmwave}

\begin{thebibliography}{10}
\providecommand{\url}[1]{#1}
\csname url@samestyle\endcsname
\providecommand{\newblock}{\relax}
\providecommand{\bibinfo}[2]{#2}
\providecommand{\BIBentrySTDinterwordspacing}{\spaceskip=0pt\relax}
\providecommand{\BIBentryALTinterwordstretchfactor}{4}
\providecommand{\BIBentryALTinterwordspacing}{\spaceskip=\fontdimen2\font plus
\BIBentryALTinterwordstretchfactor\fontdimen3\font minus
  \fontdimen4\font\relax}
\providecommand{\BIBforeignlanguage}[2]{{%
\expandafter\ifx\csname l@#1\endcsname\relax
\typeout{** WARNING: IEEEtran.bst: No hyphenation pattern has been}%
\typeout{** loaded for the language `#1'. Using the pattern for}%
\typeout{** the default language instead.}%
\else
\language=\csname l@#1\endcsname
\fi
#2}}
\providecommand{\BIBdecl}{\relax}
\BIBdecl

\bibitem{rappa:access:13}
T.~S. Rappaport, S.~Sun, R.~Mayzus, H.~Zhao, Y.~Azar, K.~Wang, G.~N. Wong,
  J.~K. Schulz, M.~Samimi, and F.~Gutierrez, ``Millimeter wave mobile
  communications for 5g cellular: {I}t will work!'' \emph{{IEEE Access}},
  vol.~1, pp. 335--349, 2013.

\bibitem{roh:commmag:14}
W.~Roh, J.~Y. Seol, J.~Park, B.~Lee, J.~Lee, Y.~Kim, J.~Cho, K.~Cheun, and
  F.~Aryanfar, ``Millimeter-wave beamforming as an enabling technology for {5G}
  cellular communications: theoretical feasibility and prototype results,''
  \emph{IEEE Comm. Mag.}, vol.~52, no.~2, pp. 106--113, Feb. 2014.

\bibitem{bai:twc:15}
T.~Bai and R.~W. Heath, ``Coverage and rate analysis for millimeter-wave
  cellular networks,'' \emph{IEEE Trans. Wireless Comm.}, vol.~14, no.~2, pp.
  1100--1114, Feb. 2015.

\bibitem{boccardi:commag:16}
F.~Boccardi, H.~Shokri-Ghadikolaei, G.~Fodor, E.~Erkip, C.~Fischione,
  M.~Kountouris, P.~Popovski, and M.~Zorzi, ``Spectrum pooling in mmwave
  networks: {O}pportunities, challenges, and enablers,'' \emph{IEEE Comm.
  Mag.}, vol.~54, no.~11, pp. 33--39, Nov. 2016.

\bibitem{shokri:jsac:16}
H.~Shokri-Ghadikolaei, F.~Boccardi, C.~Fischione, G.~Fodor, and M.~Zorzi,
  ``Spectrum sharing in mmwave cellular networks via cell association,
  coordination, and beamforming,'' \emph{IEEE Jour. Select. Areas in Comm.},
  vol.~34, no.~11, pp. 2902--2917, Nov. 2016.

\bibitem{gupta:tcom:16}
A.~K. Gupta, J.~G. Andrews, and R.~W. Heath, ``On the feasibility of sharing
  spectrum licenses in mm{W}ave cellular systems,'' \emph{IEEE Trans. Comm.},
  vol.~64, no.~9, pp. 3981--3995, Sep. 2016.

\bibitem{valenzuela:ctw:17}
R.~Valenzuela, ``{5G} technologies: Opportunities and challenges,'' in
  \emph{Proc. of IEEE Comm. Th. Workshop}, Jun. 2017.

\bibitem{kibilda:twc:17}
J.~Kibiłda, N.~J. Kaminski, and L.~A. DaSilva, ``Radio access network and
  spectrum sharing in mobile networks: {A} stochastic geometry perspective,''
  \emph{IEEE Trans. Wireless Comm.}, vol.~16, no.~4, pp. 2562--2575, Apr. 2017.

\bibitem{haykin:jsac:05}
S.~Haykin, ``Cognitive radio: brain-empowered wireless communications,''
  \emph{IEEE Jour. Select. Areas in Comm.}, vol.~23, no.~2, pp. 201--220, Feb.
  2005.

\bibitem{goldsmith:ieee:09}
A.~Goldsmith, S.~A. Jafar, I.~Maric, and S.~Srinivasa, ``Breaking spectrum
  gridlock with cognitive radios: {A}n information theoretic perspective,''
  \emph{{Proceedings of the IEEE}}, vol.~97, no.~5, pp. 894--914, May 2009.

\bibitem{gupta:jsac:16}
A.~K. Gupta, A.~Alkhateeb, J.~G. Andrews, and R.~W. Heath, ``Gains of
  restricted secondary licensing in millimeter wave cellular systems,''
  \emph{IEEE Jour. Select. Areas in Comm.}, vol.~34, no.~11, pp. 2935--2950,
  Nov. 2016.

\bibitem{rebato:infocom:16}
M.~Rebato, M.~Mezzavilla, S.~Rangan, and M.~Zorzi, ``Resource sharing in {5G}
  {mmWave} cellular networks,'' in \emph{Proc. IEEE Int. Conf. on Comp. and
  Comm. (INFOCOM) Wkshps}, Apr. 2016, pp. 271--276.

\bibitem{rebato:arxiv:16}
\BIBentryALTinterwordspacing
M.~Rebato, F.~Boccardi, M.~Mezzavilla, S.~Rangan, and M.~Zorzi, ``Hybrid
  spectrum sharing in mmwave cellular networks,'' \emph{CoRR}, 2016. [Online].
  Available: \url{http://arxiv.org/abs/1610.01339}
\BIBentrySTDinterwordspacing

\bibitem{li:crowncom:14}
G.~Li, T.~Irnich, and C.~Shi, ``Coordination context-based spectrum sharing for
  5{G} millimeter-wave networks,'' in \emph{International Conference on
  Cognitive Radio Oriented Wireless Networks and Communications (CROWNCOM)},
  Jun. 2014, pp. 32--38.

\bibitem{feng:gcwkr:14}
W.~Feng, Y.~Li, D.~Jin, and L.~Zeng, ``Inter-network spatial sharing with
  interference mitigation based on {IEEE} 802.11ad {WLAN} system,'' in
  \emph{Proc. IEEE Glob. Comm. Conf. Wkshps}, Dec. 2014, pp. 752--758.

\bibitem{fund:sarnoff:16}
F.~Fund, S.~Shahsavari, S.~S. Panwar, E.~Erkip, and S.~Rangan, ``Do open
  resources encourage entry into the millimeter wave cellular service market?''
  in \emph{{Proc. IEEE Sarnoff Symp.}}, Sep. 2016, pp. 1--2.

\bibitem{fund:arxiv_economic:17}
\BIBentryALTinterwordspacing
------, ``Spectrum and infrastructure sharing in millimeter wave cellular
  networks: {A}n economic perspective,'' \emph{CoRR}, vol. abs/1605.04602,
  2016. [Online]. Available: \url{http://arxiv.org/abs/1605.04602}
\BIBentrySTDinterwordspacing

\bibitem{nigam:tcom:14}
G.~Nigam, P.~Minero, and M.~Haenggi, ``Coordinated multipoint joint
  transmission in heterogeneous networks,'' \emph{IEEE Trans. Comm.}, vol.~62,
  no.~11, pp. 4134--4146, Nov. 2014.

\bibitem{lee:twc:15}
N.~Lee, D.~Morales-Jimenez, A.~Lozano, and R.~W. Heath, ``Spectral efficiency
  of dynamic coordinated beamforming: {A} stochastic geometry approach,''
  \emph{IEEE Trans. Wireless Comm.}, vol.~14, no.~1, pp. 230--241, Jan. 2015.

\bibitem{li:tcom:15}
C.~Li, J.~Zhang, M.~Haenggi, and K.~B. Letaief, ``User-centric intercell
  interference nulling for downlink small cell networks,'' \emph{IEEE Trans.
  Comm.}, vol.~63, no.~4, pp. 1419--1431, Apr. 2015.

\bibitem{alk:twc:15}
A.~Alkhateeb, G.~Leus, and R.~W. Heath, ``Limited feedback hybrid precoding for
  multi-user millimeter wave systems,'' \emph{IEEE Trans. Wireless Comm.},
  vol.~14, no.~11, pp. 6481--6494, Nov. 2015.

\bibitem{kulkarni:tcom:16}
M.~N. Kulkarni, A.~Ghosh, and J.~G. Andrews, ``A comparison of {MIMO}
  techniques in downlink millimeter wave cellular networks with hybrid
  beamforming,'' \emph{IEEE Trans. Comm.}, vol.~64, no.~5, pp. 1952--1967, May
  2016.

\bibitem{baccelli:inria}
F.~Baccelli and B.~Blaszczyszyn, \emph{{Stochastic Geometry and Wireless
  Networks, Volume I - Theory}}, ser. Foundations and Trends in
  Networking.\hskip 1em plus 0.5em minus 0.4em\relax {Now Publishers}, 2009,
  vol.~3.

\bibitem{sarabjot:jsac:15}
S.~Singh, M.~N. Kulkarni, A.~Ghosh, and J.~G. Andrews, ``Tractable model for
  rate in self-backhauled millimeter wave cellular networks,'' \emph{IEEE Jour.
  Select. Areas in Comm.}, vol.~33, no.~10, pp. 2196--2211, Oct. 2015.

\bibitem{alk:jstsp:14}
A.~Alkhateeb, O.~E. Ayach, G.~Leus, and R.~W. Heath, ``Channel estimation and
  hybrid precoding for millimeter wave cellular systems,'' \emph{IEEE Jour.
  Select. Topics in Sig. Proc.}, vol.~8, no.~5, pp. 831--846, Oct. 2014.

\bibitem{gao:cl:16}
Z.~Gao, C.~Hu, L.~Dai, and Z.~Wang, ``Channel estimation for millimeter-wave
  massive {MIMO} with hybrid precoding over frequency-selective fading
  channels,'' \emph{IEEE Comm. Lett.}, vol.~20, no.~6, pp. 1259--1262, Jun.
  2016.

\bibitem{renzo:twc:15}
M.~D. Renzo, ``Stochastic geometry modeling and analysis of multi-tier
  millimeter wave cellular networks,'' \emph{IEEE Trans. Wireless Comm.},
  vol.~14, no.~9, pp. 5038--5057, Sep. 2015.

\bibitem{shokri:tcom:15}
H.~Shokri-Ghadikolaei, C.~Fischione, G.~Fodor, P.~Popovski, and M.~Zorzi,
  ``Millimeter wave cellular networks: {A MAC} layer perspective,'' \emph{IEEE
  Trans. Comm.}, vol.~63, no.~10, pp. 3437--3458, Oct. 2015.

\bibitem{haenggi:tit:05}
M.~Haenggi, ``On distances in uniformly random networks,'' \emph{IEEE Trans.
  Info. Th.}, vol.~51, no.~10, pp. 3584--3586, Oct. 2005.

\bibitem{spark:tsp:17}
S.~Park, J.~Park, A.~Yazdan, and R.~Heath, ``Exploiting spatial channel
  covariance for hybrid precoding in massive {MIMO} systems,'' \emph{IEEE
  Trans. Sig. Proc.}, vol.~PP, no.~99, pp. 1--1, 2017.

\bibitem{li:twc:16}
Y.~Li, F.~Baccelli, J.~G. Andrews, T.~D. Novlan, and J.~C. Zhang, ``Modeling
  and analyzing the coexistence of {Wi-Fi} and {LTE} in unlicensed spectrum,''
  \emph{IEEE Trans. Wireless Comm.}, vol.~15, no.~9, pp. 6310--6326, Sep. 2016.

\bibitem{guidolin:gc:15}
F.~Guidolin and M.~Nekovee, ``Investigating spectrum sharing between 5{G}
  millimeter wave networks and fixed satellite systems,'' in \emph{Proc. IEEE
  Glob. Comm. Conf. Wkshps}, Dec. 2015, pp. 1--7.

\bibitem{raymond:jsac:16}
S.~S. Raymond, A.~Abubakari, and H.~S. Jo, ``Coexistence of power-controlled
  cellular networks with rotating radar,'' \emph{IEEE Jour. Select. Areas in
  Comm.}, vol.~34, no.~10, pp. 2605--2616, Oct. 2016.

\end{thebibliography}

\end{document}